\documentclass[11pt,letterpaper]{article} 
\usepackage[letterpaper,margin=1in]{geometry}
\usepackage[parfill]{parskip}
\usepackage{setspace}
\usepackage{hyperref}
\usepackage{url}
\usepackage[
  backend=biber,
  style=alphabetic,
  natbib=true,
  maxbibnames=99,
  maxalphanames=99,
  maxcitenames=5
]{biblatex}
\usepackage{todonotes}
\usepackage{amsmath}
\usepackage{amssymb}
\usepackage{amsthm}
\usepackage{thmtools}
\usepackage{amsfonts}
\usepackage{mathtools}
\usepackage{xcolor}
\usepackage{graphicx}
\usepackage{tikz}
\usetikzlibrary{patterns,3d}
\usepackage{pgfplots}
\pgfplotsset{compat=1.18}
\usepgfplotslibrary{patchplots}
\usepackage{thmtools}
\usepackage{thm-restate}
\usepackage{cleveref}

\usepackage{subcaption}
\Crefname{subfigure}{Figure}{Figures}
\usepackage{authblk}
\usepackage{algorithm}
\usepackage[noend]{algpseudocode}
\algrenewcommand\algorithmicrequire{\textbf{Input:}}
\algrenewcommand\algorithmicensure{\textbf{Output:}}

\newtheorem{theorem}{Theorem}[section]
\newtheorem{corollary}[theorem]{Corollary}
\newtheorem{lemma}[theorem]{Lemma}
\newtheorem{proposition}[theorem]{Proposition}

\theoremstyle{definition}

\newcommand{\conv}{\operatorname{conv}}

\title{Tropical Circuits with Scalar Multiplication Gates
\thanks{
This is the authors' final draft, which is to appear in the Proceedings of the 2027 Annual ACM-SIAM Symposium on Discrete Algorithms (SODA), published by the Society for Industrial and Applied Mathematics (SIAM).
Copyright \copyright\ 2027. Copyright for this paper is retained by authors.
}}

\author[1]{Christoph Hertrich}
\author[1]{Moritz Stargalla}
\affil[1]{University of Technology Nuremberg

{\normalsize\texttt{christoph.hertrich@utn.de}, \texttt{moritz.stargalla@utn.de}}}

\definecolor{OIblue}{RGB}{86,180,233}
\definecolor{OIorange}{RGB}{213,94,0}
\definecolor{OIgreen}{RGB}{0,158,115}

\newcommand{\R}{\mathbb{R}}
\newcommand{\N}{\mathbb{N}}

\DeclareMathOperator*{\xc}{xc}
\DeclareMathOperator*{\vxc}{vxc}
\DeclareMathOperator*{\nnc}{nnc}
\DeclareMathOperator*{\mnnc}{mnnc}
\DeclareMathOperator*{\size}{size}
\DeclareMathOperator*{\supp}{supp}

\begin{document}

\maketitle

\begin{abstract}
\noindent
We study tropical circuits with scalar multiplication gates, that is, algebraic circuits whose gates implement $\max$, $+$, or multiplication with a positive constant.
For such circuits, we prove exponential size lower bounds for computing maximum weight directed spanning trees and maximum weight bipartite perfect matchings. As a corollary, we obtain an exponential size separation between monotone and non-monotone maxout neural networks, which generalize the popularly used ReLU neural networks. One conclusion from this is that neural network models with enforced convexity constraints, such as \emph{input-convex neural networks} (ICNNs), sometimes need to be exponentially larger than their unrestricted counterparts in order to express the same functions.
\end{abstract}
\clearpage
\section{Introduction}

\emph{Tropical circuits}~\cite{jukna2023tropical}, also known as max-plus circuits, are a variant of classical arithmetic circuits that use maximum and addition gates instead of addition and multiplication gates. Besides fundamental interest in the power of different models of computation and their dependence on the set of allowed operations, a primary motivation to study tropical circuits is to prove lower bounds on \emph{pure dynamic programs}. By definition, a pure dynamic program consists of a predefined sequence of max (or min) and plus operations only. One example is the Bellman-Ford algorithm. Since every pure dynamic program can be written as a tropical circuit, lower bounds on the latter imply lower bounds on the former. That way, it has for example been shown that every pure dynamic program for the minimum spanning tree problem needs exponentially many iterations~\cite{jukna2019greedy}.

Recently, tropical circuits received increased attention due to their close connection to neural networks. Variants of tropical circuits have been used in order to prove size upper bounds for neural networks with \emph{rectified linear unit} (ReLU) activations~\cite{hertrich2025relu,hertrich2026arithmetic}. Such neural networks can be defined as a circuit in which each node (neuron) computes an affine function of the outputs of its predecessors composed with the ReLU function $x\mapsto\max\{0,x\}$. It is straightforward to verify that a ReLU network can exactly simulate every tropical circuit. However, ReLU networks are strictly more powerful than tropical circuits, as they can, for instance, solve the minimum spanning tree problem in polynomial size~\cite{fomin2016subtraction,hertrich2025relu}. The main reason for this distinction seems to be the ability of neural networks to implement subtraction via negative weights: the construction at hand basically implements a $(\max,+,-)$-circuit.

However, besides subtraction, there is a second feature that seemingly makes neural networks more powerful than tropical circuits: namely scalar multiplication with arbitrary real constants. This opens up the question of how much additional power this feature provides alone, without allowing subtraction. To study this question, we propose to augment the model of tropical circuits by scalar multiplication gates with positive constants, calling the resulting model \emph{scalar tropical circuits} (STCs).

Every STC computes a continuous piecewise linear (CPWL) function of the form $c\mapsto \max_{x\in P}c^\top x$ for a polytope $P$, which is the \emph{support function} $f_P$ of the polytope $P$.

\subsection{Our Contributions}
In the following we first detail our contribution in the context of tropical circuit theory and afterwards discuss implications, particularly in the context of neural networks.

\paragraph{Lower Bounds on the Size of STCs.}
The size of a regular tropical circuit is the number of $\max$ and plus gates.
We also measure the size of an STC $\Phi$, denoted as $\size(\Phi)$, as the number of $\max$ and plus gates, not counting scalar multiplication gates.
Furthermore, we use $\size_+(\Phi)$ to count only plus gates.

It is easy to see that there are functions that can be computed by STCs of smaller size than tropical circuits.
For example, $x\mapsto n x$ for $n\in \N_{\geq 2}$ can be realized by an STC of size $0$ with one scalar multiplication gate, whereas a normal tropical circuit would require $\Omega(\log n)$ plus gates.
However, for functions that have only 0-1 coefficients, that is, functions of the form $\max_{a\in A}a^\top x$ with $A\subseteq \{0,1\}^d$, the situation is less clear. Such functions are precisely support functions of 0-1 polytopes, also called \emph{multilinear tropical polynomials}.
To transfer size lower bounds from tropical circuits to STCs for a particular function, one must show that allowing arbitrary positive scalar multiplications cannot reduce the circuit size.

Interestingly, while one could expect that non-integral constants are of little use for representing support functions of 0-1 polytopes, this intuition breaks for a closely related question on neural network \emph{depth} instead of \emph{size}: Even for representing the very simple maximum function $\max\{x_1,\dots,x_n\}$, which is the support function of the standard simplex, one needs $\log_2(n)$ hidden ReLU layers \emph{if one restricts to integer weights}~\citep{haase2023lower}. However, when allowing fractional weights, this can be done with $\log_6(n)$ hidden layers~\citep{bakaev2026better,ruess2026shallower}. One motivation for us is the question whether a similar phenomenon can occur for \emph{size} instead of \emph{depth}.

A systematic way to obtain lower bounds on STCs is via \emph{extension complexity}. The extension complexity $\xc(P)$ of a polytope $P$ is the minimum number of facets of any polytope $Q$ that linearly projects to $P$. Extension complexity is a frequently studied notion in combinatorial optimization as it quantifies the minimum number of inequalities required in any linear program optimizing over $P$. It follows from~\cite{hertrich2024neural} that $\xc(P)/2$ lower-bounds the size of any STC computing $f_P$. In seminal work, \citet{fiorini2015exponential} and \citet{rothvoss2017matching} showed that, for example, the TSP polytope and the matching polytope have exponential extension complexity, which implies that STCs optimizing over those polytopes must have exponential size, too.
However, there are several polytopes with polynomial extension complexity for which there are still exponential lower bounds on tropical circuits computing their support function, e.g., the Birkhoff polytope \citep{jerrum1982some} or the (un)directed spanning tree polytope \citep{jerrum1982some,jukna2015lower,jukna2019greedy}. This leads to the following question:
\begin{center}
    \emph{Are there polytopes with polynomial extension complexity that require STCs of exponential size?}
\end{center}
We answer this question in the affirmative by giving two examples of classes of such polytopes.
The first example is the \emph{Birkhoff polytope} $P_{\rm PERM}$, which is the convex hull of all characteristic vectors of perfect matchings of the complete bipartite graph $K_{n,n}$.
The problem of evaluating the support function $f_{P_{\rm PERM}}$ is often called the \emph{assignment problem}.
Further, $f_{P_{\rm PERM}}$ is the tropical version of the permanent.
\citet{jerrum1982some} proved that any tropical circuit computing $f_{P_{\rm PERM}}$ must have at least $2^{\Omega(n)}$ plus gates (in fact, they show that at least $n(2^n-1)$ plus gates are necessary).
We extend this result by showing that any STC computing $f_{P_{\rm PERM}}$ must have at least $2^{\Omega(n)}$ plus gates.
\begin{restatable}{thm}{matchingstclowerbound}
\label{thm:matching_lower_bound}
Let $P_{\rm PERM}$ be the Birkhoff Polytope for $K_{n,n}$.
Then every STC $\Phi$ computing $f_{P_{\rm PERM}}$ satisfies $\size_+(\Phi)\in 2^{\Omega(n)}$.
\end{restatable}
In particular, this shows that the lower bound via extension complexity in~\citep{hertrich2024neural} can be exponentially loose, because $\xc(P_{\rm PERM})\in \mathcal{O}(n^2)$ \citep{fiorini2013combinatorial}.

We prove a similar result for the directed spanning tree polytope $P_{\rm DST}$, which is the convex hull of all characteristic vectors of directed spanning trees with root $n$ of the complete directed graph with $n$ nodes.
Again, this extends a result by \citet{jerrum1982some}, who showed that any tropical circuit computing $f_{P_{\rm DST}}$ has at least $2^{\Omega(n)}$ plus gates (in fact, they show that at least $(4/3)^n n^{-1}$ plus gates are necessary).
\begin{restatable}{thm}{dststclowerbound}
\label{thm:dst_lower_bound}
    Every STC $\Phi$ computing $f_{P_{\rm DST}}$ satisfies $\size_+(\Phi)\in 2^{\Omega(n)}$.
\end{restatable}

\paragraph{Proof Techniques.}
In our proofs, we generalize a widely used lower bounding technique from arithmetic and tropical circuit complexity.
There, one first proves a decomposition lemma: if a target polynomial with many monomials can be computed by a circuit of small size, then it can be written as a (tropical) sum of a small number of (tropical) products of simpler polynomials.
A lower bound is then obtained by showing that each such simple product can contain only few monomials of the target function, so many products are needed to generate all monomials of the target function.
This general technique has been widely used, for example, in \citep{hyafil1978parallel,valiant1980negation,hrubevs2011homogeneous,raz2011multilinear,jukna2015lower,jukna2019greedy,yehudayoff2019separating,srinivasan2020strongly,chattopadhyay2022monotone,kluk2025lower,cavalar2026negations}, also compare the discussions in \citep{shpilka2010arithmetic,jukna2016tropical}.
For multilinear homogeneous polynomials, the degree of the monomials in a polynomial was often (implicitly) used as a measure for how simple a polynomial is \citep{shpilka2010arithmetic,jukna2015lower,jukna2019greedy}.

The decomposition lemmas from tropical circuit complexity cannot be applied to STCs, because scalar multiplication gates can create non-integral monomials which correspond to non-integral vertices in the corresponding Newton polytopes.
In contrast, tropical circuits always yield Newton polytopes with integral vertices.
We therefore use the duality between support functions and polytopes to prove a more general, polytopal version of the decomposition step.
Namely, if the support function of a polytope $Q$ can be computed by a small STC $\Phi$, then $Q$ can be written as $Q = \mathrm{conv}\left(A_0 \cup \bigcup_{i=1}^k (A_i + B_i)\right)$ with $k\leq \size(\Phi)$, where the addition is Minkowski sum, and the polytopes $A_0,A_1,\dots,A_k$ are simple with respect to a suitable measure.

For the polytopes $P_{\rm PERM}$ and $P_{\rm DST}$ in \Cref{thm:matching_lower_bound} and \Cref{thm:dst_lower_bound}, we apply this decomposition and show that a single simple Minkowski sum $A_i + B_i$ can contain only a small fraction of the vertices of our target polytope.
For this, a crucial step is to choose a suitable measure.
For multilinear homogeneous polynomials, using the degree of a polynomial as a measure is often sufficient to prove strong lower bounds for arithmetic and tropical circuits.
The polytopal analogue of this measure is the function $P\mapsto \max_{v\in V(P)}|\supp(v)|$, where $V(P)$ is the set of vertices of $P$.
For STCs, this measure does not naturally give strong lower bounds and one has to find other suitable measures, which is a challenge in itself.
We instead use measures tailored to each problem such that $A_i$ being simple with respect to that measure forces the Minkowski sum $A_i+B_i$ to contain only a small fraction of vertices of the target polytope.
Compared to tropical circuits, this requires additional work, as the polytopes $A_i$ and $B_i$ are not restricted to have only integral vertices.
To this end, we derive structural constraints on Minkowski sums $A_i + B_i$ contained in the target polytope, such that combining these constraints with the fact that $A_i$ must be simple with respect to the chosen measure implies a bound on the number of vertices that can lie in one such Minkowski sum.
This implies that we need many such Minkowski sums, leading to lower bounds on circuit size.

\paragraph{Implications on Dynamic Programs.}
In light of the original motivation to study tropical circuits \citep{jukna2015lower,jukna2019greedy,jukna2023tropical}, \Cref{thm:matching_lower_bound} and \Cref{thm:dst_lower_bound} imply that pure dynamic programs solving the respective problems still need exponentially many iterations even if you augment them with the ability to perform positive scalar multiplications.

\paragraph{Implications on Neural Networks.}
Our results imply a novel exponential size separation between monotone and non-monotone neural networks. Before we explain this separation result in detail, we first introduce maxout networks as a generalization of ReLU networks that is mathematically cleaner to analyze.

A \emph{rank-$k$ maxout network} is a neural network in which each neuron computes the maximum of $k$ affine functions of the outputs of its predecessors, instead of just computing the maximum of zero with a single affine function as in ReLU networks.
The size of a maxout network is the number of maxout neurons.
For every fixed $k\in \N_{\geq 2}$, every rank-$k$ maxout network can be simulated via a ReLU network with a constant multiplicative size overhead.

A neural network is called \emph{monotone} if it does not contain any negative weight~\citep{daniels2010monotone}. As a result, a monotone (ReLU or maxout) neural network computes a \emph{convex} and \emph{monotone} function. If one instead allows negative weights on outgoing connections of input neurons, but nowhere else in the network, one still ensures convexity, even if monotonicity might be lost. Such architectures are known as \emph{input-convex neural networks} (ICNNs)~\citep{amos2017input}. Disallowing negative weights can be beneficial for a variety of reasons. From a practical viewpoint, enforced monotonicity or convexity are an effective way to incorporate prior knowledge into the network architecture and make the model more interpretable~\citep{zimmermann2026hidden}. ICNNs have gained quite some popularity in the machine learning community for such reasons~\citep{amos2017input,chen2018optimal,makkuva2020optimal,huang2021convex}, even though their performance is often significantly worse than with unrestricted models~\citep{gagneux2025convexity}. From a theoretical viewpoint, monotone models provide a restricted setting in which lower bounds can be proven more easily compared to the unrestricted model~\citep{mikulincer2024size,valerdi2024minimal,bakaev2025depth,hertrich2024neural}, similarly to what is often done for related models of computation like Boolean or arithmetic circuits~\citep{valiant1980negation,alon1987monotone,shpilka2010arithmetic}.
Note that, while in the non-monotone case ReLU and maxout networks are essentially equivalent, not every monotone maxout network can be simulated by a monotone ReLU network~\citep{bakaev2025depth}. Therefore, proving lower bounds on monotone maxout networks is potentially stronger than proving the same bounds for monotone ReLU networks.

In order to formalize complexity statements about monotone and general neural networks, \citet{hertrich2024neural} define the \emph{neural network complexity} $\nnc(P)$ of a polytope $P$ as the minimum size of a rank-$2$ maxout network computing the support function $f_P$; and the \emph{monotone neural network complexity} $\mnnc(P)$ as the minimum size of a rank-$2$ maxout ICNN computing $f_P$. If $P$ is contained in the nonnegative orthant, then they show that $\mnnc(P)$ equals the minimum size of a monotone maxout network computing $f_P$, justifying the name $\mnnc(P)$ even though it is defined via the more general ICNN model.

As for STCs, lower bounds on ICNNs can be proved via extension complexity. More precisely, \citet{hertrich2024neural} show that $\xc(P)\leq2\cdot\mnnc(P)$ for every polytope~$P$. Combined with the result by \citet{rothvoss2017matching}, this implies an exponential lower bound on $\mnnc(P_{\rm M})$ for the matching polytope $P_{\rm M}$. However, although $f_{P_{\rm M}}$ can be computed in polynomial time \citep{edmonds1965paths}, it is unknown whether it can be exactly computed by a neural network of polynomial size, compare the discussion in \citep{hertrich2025relu}.

Overall, the following question has been open so far.

\begin{quote}
    \emph{Are there functions representable with polynomial size neural networks for which maxout ICNNs need exponential size?}
\end{quote}

We answer this question in the affirmative.
More precisely, we obtain the following as corollaries of \Cref{thm:matching_lower_bound} and \Cref{thm:dst_lower_bound}.

\begin{restatable}{cor}{matchingmnnc}
\label{cor:mnnc_matching_lower_bound}
    Let $P_{\rm PERM}$ be the Birkhoff Polytope for $K_{n,n}$.
    Then $\mnnc(P_{\rm PERM})\in 2^{\Omega(n)}$.
\end{restatable}
This again shows that the lower bound based on extension complexity~\citep{hertrich2024neural} can be exponentially loose, since $\xc(P_{\rm PERM})\in \mathcal{O}(n^2)$.

\begin{restatable}{cor}{dstseparation}
\label{cor:mnnc_dst_lower_bound}
    Let $P_{{\rm DST}}$ be the directed spanning tree polytope on $n$ vertices.
    Then, it holds that $\nnc(P_{\rm DST})\in \mathcal{O}(n^3)$ and $\mnnc(P_{\rm DST})\in 2^{\Omega(n)}$.
\end{restatable}
This answers the above question and shows that maxout ICNNs must sometimes be exponentially larger than their unrestricted counterparts to express the same functions.

We obtain the corollaries by showing that STCs can simulate bias-free monotone maxout networks with small overhead and combining this with our lower bounds for STCs.
The upper bound in \Cref{cor:mnnc_dst_lower_bound} follows from tropicalizing a polynomial-size subtraction-free $(+,\times, /)$-circuit of \citep{fomin2016subtraction} computing the arithmetic version of the polynomial $f_{P_{\rm DST}}$, which gives a polynomial-size $(\max,+, -)$-circuit computing $f_{P_{\rm DST}}$.

\subsection{Further Related Work}
Our lower bounds are closely related to the literature on tropical and monotone arithmetic circuits.
For a survey on arithmetic circuit complexity, see \citep{shpilka2010arithmetic}.
\citep{jerrum1982some} related the tropical circuit complexity of homogeneous multilinear polynomials to monotone arithmetic circuit complexity and proved exponential lower bounds for several polynomials corresponding to combinatorial optimization problems, including the tropical permanent, TSP, and the maximum weight directed spanning tree problem.
Subsequent works used tropical circuits as a model for pure dynamic programming and proved lower bounds for various settings \citep{jukna2015lower,jukna2016tropical,jukna2019greedy,jukna2023tropical,kluk2025lower}.
Interestingly, while many of those lower bounds for tropical circuits are proved via the non-tropical counterparts, it is less meaningful to de-tropicalize scalar multiplication gates. The reason is that in the tropical world, the ``freshman's dream'' is true, that is, $(x_1\oplus x_2)^{\odot\alpha}=\alpha \max(x_1,x_2)=\max(\alpha x_1,\alpha x_2)=x_1^{\odot\alpha}\oplus x_2^{\odot\alpha}$ for $\alpha\in \R_{>0}$, while there is no natural way to distribute exponents over several additions in the non-tropical world for $(x_1+x_2)^\alpha$.

More generally, the effect of extending circuit models with additional operations has been studied, for instance, in the setting of adding subtraction or division to arithmetic $(+,\times)$-circuits \citep{valiant1980negation,fomin2016subtraction}, subtraction to tropical $(\max,+)$-circuits \citep{jukna2023tropical}, and negation to boolean $(\lor,\land)$-circuits \citep{razborov1985lower,alon1987monotone,tardos1988gap}.

Another line of work studies the expressivity of ReLU and maxout networks through the viewpoint of polyhedral and tropical geometry; see also the survey \citep{huchette2026deep}.
ReLU networks can be described as tropical rational functions \citep{zhang2018tropical}.
Every CPWL function can be exactly represented by a ReLU network \citep{arora2016understanding}, and it is a prominent open question if constant depth is sufficient \citep{hertrich2023towards,haase2023lower,averkov2025expressiveness,grillo2026depth,bakaev2026better,ruess2026shallower}.
In contrast, there are families of functions that require monotone networks and ICNNs of unbounded depth \citep{valerdi2024minimal,bakaev2025depth}.
Every convex CPWL function can be computed by a ReLU ICNN \citep{chen2018optimal}.
However, there are convex monotone CPWL functions such as $\max\{x_1,\dots,x_n\}$ that cannot be computed or even be approximated by a monotone ReLU network \citep{mikulincer2024size,bakaev2025depth}.
In contrast, every convex monotone CPWL function can be computed by a monotone maxout network. Furthermore, ICNNs sometimes require strictly more depth than their unrestricted counterparts~\citep{gagneux2025convexity,bakaev2025depth}. In terms of size, $\mnnc(P)$ can be lower-bounded via $\xc(P)$, as discussed earlier. In addition, $\nnc(P)$ can be lower bounded by a stronger version of $\xc(P)$, called virtual extension complexity $\vxc(P)$, quantifying the minimum number of linear inequalities required to formulate $f_P$ as a difference of two linear programs~\citep{hertrich2024neural}.
The extension complexity was related to the size of arithmetic and tropical circuits in \citep{hrubevs2023shadows}.

Our paper studies lower bounds on the required size of ReLU networks.
Upper bounds on the required size of ReLU networks have been studied for the knapsack problem in~\citep{hertrich2023provably}, for maximum flows and maximum (undirected) spanning trees in \citep{hertrich2025relu}, and for regular matroids in \citep{hertrich2026arithmetic}.
In particular, the latter result gives a first example of polytopes $P_n$ where the best known upper bound on $\vxc(P_n)\in \mathcal{O}(n^3)$ is lower than the best known upper bound on $\xc(P_n)\in \mathcal{O}(n^6)$.

\subsection{Outlook}
We initiate the study of tropical circuits with scalar multiplication gates and prove in two concrete cases that such gates do not significantly increase the power of standard tropical circuits.
The natural next step would be to extend our results to more cases or even prove a generalization: are there any 0-1 polytopes for which STCs can be more efficient in terms of size than standard tropical circuits? One way towards resolving this question negatively would be to search for general conditions that allow to transfer lower bounds on tropical circuits to STCs.

A particularly intriguing open case is that of undirected spanning trees. Even for ordinary tropical circuits, it took several decades until the exponential lower bound in the directed case by \citet{jerrum1982some} was extended to the undirected case by \citet{jukna2019greedy}.
While we do not believe that polynomial-size STCs can compute the support function of the undirected spanning tree polytope, extending the lower bound by \citet{jukna2019greedy} to STCs seems difficult as their techniques seem to rely on combinatorial structure that does not necessarily persist when allowing scalar multiplications.

In the context of neural networks, our results imply an exponential separation between general networks and their monotone / input-convex counterparts for one concrete example, namely directed spanning trees. In the context of bipartite perfect matchings, we also obtain an exponential lower bound on $\mnnc(P_{\rm PERM})$, but it remains unknown whether $\nnc(P_{\rm PERM})$ is polynomial. Closely related, \citet[Section 6.5, Problem 3]{jukna2023tropical} explicitly states the open problem whether there are polynomial-size $(\max,+,-)$-circuits to compute $f_{P_{\rm PERM}}$. This is similarly in spirit to the famous open question whether there are (non-tropical) polynomial-size arithmetic circuits to compute the permanent, equivalently, whether $\mathrm{VP}=\mathrm{VNP}$, but it is more reasonable to expect a polynomial upper bound in the tropical case. An indication for this is that the tropical permanent can be computed in polynomial time via the Hungarian method, while evaluating the non-tropical permanent is $\mathrm{\#P}$-complete. However, polynomial-time computability does not imply the existence of polynomial-size $(\max,+,-)$-circuits or neural networks, compare~\citep{hertrich2025relu}, and it remains an open problem whether there exists a class of CPWL functions that is computable in polynomial time, but not representable with polynomial-size neural networks~\citep{hertrich2024neural}. A candidate for such a function class is the class of support functions of the (non-bipartite) matching polytope, as it has exponential extension complexity~\citep{rothvoss2017matching}, but can still be evaluated in polynomial time~\citep{edmonds1965paths}.

More generally, our first example of an exponential separation between $\mnnc$ and $\nnc$ opens up the broader mission to figure out whether this is a singular phenomenon or whether there are many functions for which monotone networks are significantly less powerful than non-monotone variants. In other words, what features make a function difficult to represent for monotone networks?

In this paper we focus on a real model of computation and require exact representations over real numbers. Especially in the practice of neural networks, approximate representations are often sufficient. It would therefore be interesting to study size lower bounds for approximating support functions with STCs.
Approximation lower bounds have been studied for tropical circuits in \citep{jukna2020approximation}, and for maxout ICNNs from the extension-complexity viewpoint in \citep{hertrich2024neural}.

Finally, while we prove exponential lower bounds for STCs, the constants in the exponents of our lower bounds are weaker than in the related lower bounds on tropical circuits. It would therefore be interesting to close this gap. More precisely, dropping polynomial factors, our lower bound for $P_{\rm PERM}$ is asymptotically $2^{cn}$ for $c\approx0.918$ compared to $2^n$ in \citep{jerrum1982some}; for $P_{\rm DST}$, our lower bound is $(9/7)^n$ compared to $(4/3)^n$ in \citep{jerrum1982some}.

\section{Preliminaries}
\paragraph{Notation.}
We define $\N=\{0,1,\dots\}$ as the set of natural numbers including zero.
For $n\in \N_{\geq 1}$, we write $[n]\coloneq \{1,\dots,n\}$.
We denote the standard basis vectors of $\R^d$ by $\{e_1,\dots, e_d\}$.
For two sets $P,Q\subset \R^d$, their \emph{Minkowski sum} is $P+Q=\{p+q: p\in P, q\in Q\}$.
For $z\in \R^d$, we write $\supp(z)=\{i\in [d]: z_i\neq 0\}$.
For $\lambda \in \R_{> 0}$ and $P\subseteq \R^d$, the \emph{dilation} of $P$ by $\lambda$ is $\lambda P\coloneq \{\lambda p: p\in P\}$.
A function $f:\R^d\to \R$ is \emph{monotone} if $f(x)\leq f(y)$ for all $x,y\in \R^d$ with $x\leq y$ component-wise.
It is \emph{positively homogeneous} if $f(\lambda x)=\lambda f(x)$ for all $\lambda\in \R_{\geq 0}$.

\paragraph{Polyhedra and Extension Complexity.}
A \emph{polyhedron} $P$ is the intersection of finitely many closed halfspaces $P=\{x\in \R^d: Ax \leq b\}$.
A \emph{face} of $P$ is either the empty set or the set of maximizers $\arg\max\{c^\top x: x\in P\}$ of a linear function over $P$.
Faces of dimension zero are called \emph{vertices}; faces of dimension $\dim(P)-1$ are called \emph{facets}.
A \emph{polytope} is a bounded polyhedron.
By the Minkowski-Weyl theorem, $P$ can be equivalently written as a convex hull of finitely many points.
The inclusion-wise minimal set $V$ with $\conv(V)=P$ is the set of vertices $V(P)$ of $P$.
The \emph{extension complexity} $\xc(P)$ of a polytope $P$ is the minimum number of facets of a polytope, possibly in a higher-dimensional space, that affinely projects to $P$.

\paragraph{Support Functions and Duality.}
For a polytope $P\subset \R^d$, the \emph{support function} of $P$ is $f_P:\R^d\to \R,\;c\mapsto \max_{x\in P}c^\top x$.
Support functions are convex, continuous piecewise-linear (CPWL), and positively homogeneous.
The support function $f_P$ is monotone if and only if $P\subset \R^d_{\geq 0}$.
Let $\mathcal{F}^d$ be the set of positively homogeneous convex CPWL functions from $\R^d$ to $\R$, and let $\mathcal{F}_+^d$ be the subset of monotone functions in $\mathcal{F}^d$.
There is a standard duality between $\mathcal{F}^d$ and the set $\mathcal{P}^d$ of polytopes in $\R^d$.
The map $\varphi:\mathcal{P}^d\to \mathcal{F}^d,\, P\mapsto f_P$ is a bijection satisfying $\varphi(P+Q)=f_{P+Q}=f_P+f_Q$, $\varphi(\conv(P\cup Q))=f_{\conv(P\cup Q)}=\max\{f_P,f_Q\}$, and $\varphi(\lambda P)=f_{\lambda P}=\lambda f_P$ for all $P,Q\in \mathcal{P}^d$ and all $\lambda\in \R_{\geq 0}$.
The inverse map $\varphi^{-1}:\mathcal{F}^d\to\mathcal{P}^d$ maps a function $f\in \mathcal{F}^d$ to the unique polytope whose support function is $f$.
We call this polytope the \emph{Newton Polytope} of $f$.
Restricting $\varphi$ gives a bijection between $\mathcal{F}_+^d$ and the set $\mathcal{P}_+^d$ of polytopes in $\R^d_{\geq 0}$.
See \Cref{fig:duality} for an illustration.

\begin{figure}
    \centering\resizebox{.5\linewidth}{!}{
	\begin{tikzpicture}[scale=0.7,every node/.style={transform shape}]
	\begin{scope}[shift={(-10,.2)}, scale=1.45, transform shape]
		\begin{axis}[
			anchor=center,
			view={-25}{40},
			xlabel={$x_1$},
			ylabel={$x_2$},
			xmin=-1.5, xmax=1.5,
			ymin=-1.5, ymax=1.5,
			zmin=-0.2, zmax=1.5,
			ticks=none,
			width=5cm,
			height=3.75cm,
			]
			\addplot3[draw=none, fill=OIblue!60] coordinates {
				(-1.5, 0, 0)
				(0, 0, 0)
				(1.5, 1.5, 1.5)
				(-1.5, 1.5, 1.5)
			};
			\addplot3[draw=none, fill=OIorange!60] coordinates {
				(-1.5, -1.5, 0)
				(0, -1.5, 0)
				(0, 0, 0)
				(-1.5, 0, 0)
			};
			\addplot3[draw=none, fill=OIgreen!60] coordinates {
				(0, -1.5, 0)
				(1.5, -1.5, 1.5)
				(1.5, 1.5, 1.5)
				(0, 0, 0)
			};
			\addplot3[black, line width=1.5pt] coordinates { (0,0,0) (1.5,1.5,1.5) };
			\addplot3[black, line width=1.5pt] coordinates { (0,0,0) (0,-1.5,0) };
			\addplot3[black, line width=1.5pt] coordinates { (0,0,0) (-1.5,0,0) };
		\end{axis}
		\node at (-.2,-1.5) {\large$\max\{\textcolor{OIorange!90!black}{0},\textcolor{OIgreen!90!black}{x_1},\textcolor{OIblue!90!black}{x_2}\}$};
		\node at (-.3,-.7) {\large$\textcolor{OIorange!90!black}{0}$};
		\node at (.85,.15) {\large$\textcolor{OIgreen!90!black}{x_1}$};
		\node at (-.6,.15) {\large$\textcolor{OIblue!90!black}{x_2}$};
	\end{scope}
	\draw[->, thick, bend right=20] (-2.4,1) to node[above, yshift=.1cm] {\Large Support Function} (-6.2,1);
	\draw[->, thick, bend right=20] (-6.2,-1) to node[below, yshift=-.1cm] {\Large Newton Polytope} (-2.4,-1);
	\node at (-4.3,0) {\LARGE \textbf{Duality}};
	\begin{scope}[scale=1.25, transform shape, shift={(0.5,0)}]
		\begin{axis}[
			x=1cm, y=1cm,
			xmin=-2, xmax=2,
			ymin=-2, ymax=2,
			hide axis,
			anchor=origin,
			at={(0,0)}
			]
		\end{axis}
	\node[draw, color=OIorange, circle, fill=OIorange, inner sep=2, label=270:{\large\textcolor{OIorange!90!black}{(0,0)}}] (0) at (-1,-1) {};
	\node[draw, color=OIblue, circle, fill=OIblue, inner sep=2, label=90:{\large\textcolor{OIblue!90!black}{(0,1)}}] (2) at (-1,1) {};
	\node[draw, color=OIgreen, circle, fill=OIgreen, inner sep=2, label=270:{\large\textcolor{OIgreen!90!black}{(1,0)}}] (1) at (1,-1) {};
	\draw[thick, color=black] (0) -- (1) -- (2) -- (0);
	\end{scope}
	\node at (-9.75,2.5) {\Large \textbf{CPWL Functions}};
	\node at (0.25,2.5) {\Large \textbf{Polytopes}};
    \end{tikzpicture}}
    \caption{Illustration of the duality between a CPWL function and a polytope.}
    \label{fig:duality}
\end{figure}
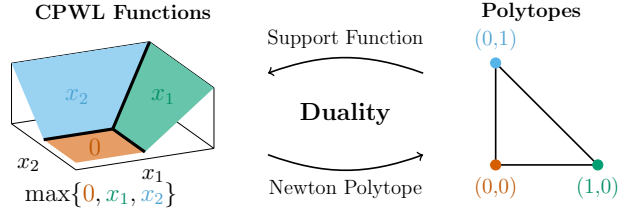
\paragraph{Maxout Networks.}
We use the definition of maxout networks from \citep{hertrich2024neural}.
A \emph{rank-$k$ maxout neural network}, where $k\geq 2$, is given by a directed acyclic graph $(V, A)$.
The $d\geq 1$ nodes of in-degree zero are called \emph{input neurons}; all other $s\geq 1$ nodes are \emph{maxout neurons}.
We assume that among the maxout neurons there is a unique \emph{output neuron} of out-degree zero.
Each node $v$ computes a function $z_v:\mathbb{R}^d\to \mathbb{R}$.
The $i$-th input node computes $z_v(x)=x_i$.
Each maxout neuron $v$ computes the maximum of $k$ affine functions of the outputs of its incoming neighbors $\delta^-_v$:
\[
z_v(x)=\max_{i=1,\dots, k}\left(\sum_{u\in \delta^{-}_v}w^i_{uv}z_u(x)+b^i_v\right),
\]
where $w^i_{uv}\in \mathbb{R}$ for $i=1,\dots, k$ are \emph{weights} of the arc $uv\in A$ and $b^i_v$ for $i=1,\dots, k$ are \emph{biases} of the node $v\in V$.
Maxout networks contain ReLU activations as a special case.
The \emph{size} of a maxout network is the number of maxout neurons.
A maxout network is \emph{monotone} if all weights are nonnegative.
It is \emph{input-convex} if negative weights are allowed only on arcs leaving an input neuron.
Monotone maxout networks compute monotone and convex CPWL functions; maxout ICNNs compute convex CPWL functions.
For a polytope $P$, the \emph{neural network complexity} $\nnc(P)$ is the minimum size of any rank-2 maxout network computing $f_P$.
The \emph{monotone neural network complexity} $\mnnc(P)$ is the minimum size of a rank-2 maxout ICNN computing $f_P$.
If $P\in\mathcal{P}_+^d$, then this coincides with the minimum size of a monotone rank-2 maxout network computing $f_P$, see \Cref{lem:ICNN_to_mon_for_mon_functions} below.

\paragraph{Tropical Circuits and Tropical Polynomials.}
A \emph{tropical circuit} $\Phi$ is a directed acyclic graph, where parallel arcs are allowed.
Each node, also called a \emph{gate}, is of one of the following types.
A gate of in-degree zero is an \emph{input gate} and holds either one of the variables $x_1,\dots, x_d$ or the constant $0$.
Every other gate has in-degree two and computes either the maximum or the sum of the values computed at its two predecessor gates.
We assume that there is a unique \emph{output gate} of out-degree zero.
The function computed at the output gate is denoted by $f_\Phi:\R^d\to \R$.
The \emph{size} of $\Phi$ is the number of non-input gates.
See \Cref{fig:tropical_circuit_example} for an illustration.

Every tropical circuit computes a \emph{tropical polynomial} $f_\Phi(x)=\max_{a\in A}a^\top x$ for some finite set $A\subseteq \N^d$.
Conversely, every function of this form can be computed by a tropical circuit.
Thus, tropical circuits compute exactly those functions in $\mathcal{F}_+^d$ whose Newton polytopes are of the form $\conv(A)$ with $A\subseteq \N^d$ finite.
A term $a^\top x$ for $a\in A$ is called a \emph{tropical monomial}.
The \emph{degree} of $a\in \N^d$ is $a_1+\dots+a_d$.
A tropical polynomial is \emph{homogeneous} if all its monomials have the same degree, and it is \emph{multilinear} if all coefficient vectors of monomials lie in $\{0,1\}^d$.

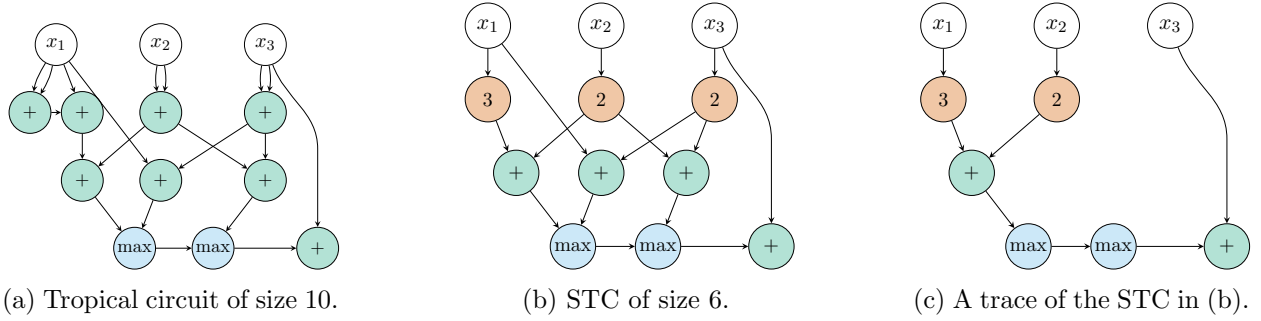
\begin{figure}
\centering

\begin{subfigure}{0.27\textwidth}
\centering
\resizebox{\linewidth}{!}{
\begin{tikzpicture}[>=stealth]
\tikzset{
    inp/.style={draw,circle,minimum size=8mm,inner sep=0pt,font=\normalsize},
    plus/.style={draw,circle,minimum size=8mm,inner sep=0pt,font=\small, fill=OIgreen!30},
    maxg/.style={draw,circle,minimum width=8mm,inner sep=0pt,font=\small, fill=OIblue!30},
    expr/.style={font=\scriptsize, align=center}
}

\node[inp] (x1) at (1,0) {$x_1$};
\node[inp] (x2) at (3,0) {$x_2$};
\node[inp] (x3) at (5,0) {$x_3$};

\node[plus] (a1) at (.5,-1.3) {$+$};      
\node[plus] (a2) at (1.5,-1.3) {$+$};      

\node[plus] (b1) at (3,-1.3) {$+$};      
\node[plus] (c1) at (5,-1.3) {$+$};      

\node[plus] (s1) at (1.5,-2.6) {$+$};    
\node[plus] (s2) at (3,-2.6) {$+$};    
\node[plus] (s3) at (5,-2.6) {$+$};    

\node[maxg] (m1) at (2.5,-3.9) {$\max$};
\node[maxg] (m2) at (4,-3.9) {$\max$};

\node[plus] (f) at (6,-3.9) {$+$};

\draw[->] (x1) to[out=240,in=80] (a1);
\draw[->] (x1) to[out=260,in=60] (a1);
\draw[->] (a1) -- (a2);
\draw[->] (x1) -- (a2);

\draw[->] (x2) to[out=260,in=100] (b1);
\draw[->] (x2) to[out=280,in=80] (b1);

\draw[->] (x3) to[out=260,in=100] (c1);
\draw[->] (x3) to[out=280,in=80] (c1);

\draw[->] (a2) -- (s1);
\draw[->] (b1) -- (s1);

\draw[->] (x1) -- (s2);
\draw[->] (c1) -- (s2);

\draw[->] (b1) -- (s3);
\draw[->] (c1) -- (s3);

\draw[->] (s1) -- (m1);
\draw[->] (s2) -- (m1);
\draw[->] (m1) -- (m2);
\draw[->] (s3) -- (m2);
\draw[->] (m2) -- (f);
\draw[->] (x3) to[out=290,in=90] (6,-2) -- (f);
\end{tikzpicture}}
\caption{Tropical circuit of size 10.}
\label{fig:tropical_circuit_example}
\end{subfigure}
\hfill
\begin{subfigure}{0.27\textwidth}
\centering
\resizebox{\linewidth}{!}{
\begin{tikzpicture}[>=stealth]
\tikzset{
    inp/.style={draw,circle,minimum size=8mm,inner sep=0pt,font=\normalsize},
    scalar/.style={draw,circle,minimum size=8mm,inner sep=0pt,font=\small, fill=OIorange!35},
    plus/.style={draw,circle,minimum size=8mm,inner sep=0pt,font=\small, fill=OIgreen!30},
    maxg/.style={draw,circle,minimum width=8mm,inner sep=0pt,font=\small, fill=OIblue!30},
    expr/.style={font=\scriptsize, align=center}
}

\node[inp] (x1) at (1,0) {$x_1$};
\node[inp] (x2) at (3,0) {$x_2$};
\node[inp] (x3) at (5,0) {$x_3$};

\node[scalar] (s31) at (1,-1.3) {$3$};
\node[scalar] (s22) at (3,-1.3) {$2$};
\node[scalar] (s23) at (5,-1.3) {$2$};

\node[plus] (p1) at (1.5,-2.6) {$+$};    
\node[plus] (p2) at (3,-2.6) {$+$};    
\node[plus] (p3) at (4.5,-2.6) {$+$};    

\node[maxg] (m1) at (2.5,-3.9) {$\max$};
\node[maxg] (m2) at (4,-3.9) {$\max$};

\node[plus] (f) at (6,-3.9) {$+$};

\draw[->] (x1) -- (s31);
\draw[->] (x2) -- (s22);
\draw[->] (x3) -- (s23);

\draw[->] (s31) -- (p1);
\draw[->] (s22) -- (p1);

\draw[->] (x1) -- (p2);
\draw[->] (s23) -- (p2);

\draw[->] (s22) -- (p3);
\draw[->] (s23) -- (p3);

\draw[->] (p1) -- (m1);
\draw[->] (p2) -- (m1);
\draw[->] (m1) -- (m2);
\draw[->] (p3) -- (m2);
\draw[->] (m2) -- (f);
\draw[->] (x3) to[out=290,in=90] (6,-2) -- (f);
\end{tikzpicture}}
\caption{STC of size 6.}
\label{fig:stc_example}
\end{subfigure}
\hfill
\begin{subfigure}{0.27\textwidth}
\centering
\resizebox{\linewidth}{!}{
\begin{tikzpicture}[>=stealth]
\tikzset{
    inp/.style={draw,circle,minimum size=8mm,inner sep=0pt,font=\normalsize},
    scalar/.style={draw,circle,minimum size=8mm,inner sep=0pt,font=\small, fill=OIorange!35},
    plus/.style={draw,circle,minimum size=8mm,inner sep=0pt,font=\small, fill=OIgreen!30},
    maxg/.style={draw,circle,minimum width=8mm,inner sep=0pt,font=\small, fill=OIblue!30},
    expr/.style={font=\scriptsize, align=center}
}

\node[inp] (x1) at (1,0) {$x_1$};
\node[inp] (x2) at (3,0) {$x_2$};
\node[inp] (x3) at (5,0) {$x_3$};

\node[scalar] (s31) at (1,-1.3) {$3$};
\node[scalar] (s22) at (3,-1.3) {$2$};

\node[plus] (p1) at (1.5,-2.6) {$+$};    

\node[maxg] (m1) at (2.5,-3.9) {$\max$};
\node[maxg] (m2) at (4,-3.9) {$\max$};

\node[plus] (f) at (6,-3.9) {$+$};

\draw[->] (x1) -- (s31);
\draw[->] (x2) -- (s22);

\draw[->] (s31) -- (p1);
\draw[->] (s22) -- (p1);

\draw[->] (p1) -- (m1);
\draw[->] (m1) -- (m2);
\draw[->] (m2) -- (f);
\draw[->] (x3) to[out=290,in=90] (6,-2) -- (f);
\end{tikzpicture}}
\caption{A trace of the STC in (b).}
\label{fig:trace_example}
\end{subfigure}
\caption{Circuits in (a) and (b) compute $\max(3x_1+2x_2+x_3,x_1+3x_3,2x_2+3x_3)$.
The numbers in the orange scalar gates are the scalars.
The trace in (c) corresponds to the monomial $3x_1+2x_2+x_3$.
}
\label{fig:two_circuits}
\end{figure}

\paragraph{Scalar Tropical Circuits.}
A \emph{scalar tropical circuit} (STC) $\Phi$ is a tropical circuit with one additional type of gate.
A \emph{scalar gate} has in-degree one, is associated with a scalar $\lambda\in \mathbb{R}_{>0}$, and maps an input value $z$ to $\lambda z$.
For an STC $\Phi$, its size $\size(\Phi)$ is the total number of $\max$ and $+$ gates.
Scalar gates are not counted in $\size(\Phi)$.
We also write $\size_+(\Phi)$ for the number of $+$ gates.
See \Cref{fig:stc_example} for an example.

The duality between polytopes and support functions gives a useful geometric interpretation of STCs.
For every gate $v$, let $f_v:\R^d\to \R$ be the function computed at $v$.
We associate with $v$ a polytope $P_v\subset\R^d_{\geq 0}$ such that $f_v=f_{P_v}$.
If $v$ is an input gate holding $x_i$, then $P_v=\{e_i\}$.
If $v$ is an input gate holding $0$, then $P_v= \{0\}$.
If $v$ is a $\max$ gate with predecessor gates $u,w$ then $P_v= \conv(P_u\cup P_w)$.
If $v$ is a $+$ gate with predecessor gates $u,w$, then $P_v= P_u+P_w$.
If $v$ is a scalar gate with predecessor $u$ and scalar $\lambda\in \mathbb{R}_{>0}$, then $P_v= \lambda P_u$.
We denote the polytope associated with the output gate by $P_\Phi$.
Then, $\Phi$ computes the support function $f_{P_\Phi}$.
This interpretation is illustrated in \Cref{fig:polytopal_viewpoint}.

It follows that STCs compute functions in $\mathcal{F}^d_+$.
Conversely, every function in $\mathcal{F}^d_+$ can be computed by a STC: each vertex of a polytope in $\mathcal{P}^d_+$ can be built from dilations of standard basis vectors and the zero vector using $+$ gates, and the convex hull of these vertices can then be built using $\max$ gates.
Thus, one can equivalently view an STC as a circuit whose input gates hold the sets $\{0\},\{e_1\},\dots,\{e_d\}$ and whose non-input gates compute binary Minkowski sums, binary convex hulls of unions, and dilations.
We will switch between these views without further comment.
This interpretation generalizes \emph{Minkowski circuits} as defined by \citet{jukna2016tropical}, which underlie several lower bound results for $(+,\times)$ arithmetic and tropical circuits; see also the discussion in \citep[p. 2065]{jukna2016tropical}.
Related polytopal viewpoints have also been used in the study of ReLU and maxout networks, for example in \citep{hertrich2023towards,balakin2025maxout}.

\begin{figure}
    \centering
    \resizebox{\linewidth}{!}{
    \begin{tikzpicture}[>=stealth]
    \tikzset{
        nodescale/.style={scale=0.9, transform shape},
        inp/.style={nodescale,draw,circle,minimum size=7mm,inner sep=0pt,font=\footnotesize},
        scalar/.style={nodescale,draw,circle,minimum size=7mm,inner sep=0pt,font=\scriptsize, fill=OIorange!35},
        plus/.style={nodescale,draw,circle,minimum size=7mm,inner sep=0pt,font=\scriptsize, fill=OIgreen!30},
        maxg/.style={nodescale,draw,circle,minimum size=7mm,inner sep=0pt,font=\scriptsize, fill=OIblue!30},
        expr/.style={font=\scriptsize, align=center},
        coord/.style={font=\fontsize{5}{6}\selectfont}
    }
    
    \node[inp] (x1) at (0,4.5) {$x_1$};
    \node[inp] (x2) at (0,3.5) {$x_2$};
    
    \node[plus] (a1) at (1.75,4) {$+$};             
    \node[maxg] (a2) at (3.5,4.5) {$\max$};      
    \node[maxg] (a3) at (5.25,3.5) {$\max$};        
    \node[plus] (a4) at (7,4) {$+$};           
    \node[maxg] (a5) at (8.75,4.5) {$\max$};        
    \node[maxg] (a6) at (11,3.5) {$\max$};      
    \node[scalar] (a7) at (13.4,4) {$2$};      
    
    \draw[->] (x1) -- (a1);
    \draw[->] (x2) -- (a1);
    \draw[->] (x1) -- (a2);
    \draw[->] (a1) -- (a2);
    \draw[->] (x2) -- (a3);
    \draw[->] (a2) -- (a3);
    \draw[->] (a1) -- (a4);
    \draw[->] (a3) -- (a4);
    \draw[->] (x1) to[out=9,in=-189] (a5);
    \draw[->] (a4) -- (a5);
    \draw[->] (x2) to[out=-7,in=-173] (a6);
    \draw[->] (a5) -- (a6);
    \draw[->] (a6) -- (a7);

    \begin{scope}[shift={(0,0.3)},scale=0.95,transform shape]
        \draw[black!20,->] (-.5,1.5) -- (.5,1.5);
        \draw[black!20,->] (-.5,1.5) -- (-.5,2.5);
        \draw[fill=black] (-.5,1.9) circle (.02cm) node[left,coord] {$(0,1)$};
        \draw[fill=black] (-.1,1.5) circle (.02cm) node[below,coord] {$(1,0)$};
    \end{scope}
    \begin{scope}[shift={(1.75,0.3)},scale=0.95,transform shape]
        \draw[black!20,->] (-.5,1.5) -- (.5,1.5);
        \draw[black!20,->] (-.5,1.5) -- (-.5,2.5);
        \draw[fill=black] (-.1,1.9) circle (.02cm) node[above,coord] {$(1,1)$};
    \end{scope}
    \begin{scope}[shift={(3.5,0.3)},scale=0.95,transform shape]
        \draw[black!20,->] (-.5,1.5) -- (.5,1.5);
        \draw[black!20,->] (-.5,1.5) -- (-.5,2.5);
        \draw[fill=black] (-.1,1.5) circle (.02cm) node[below,coord] {$(1,0)$};
        \draw[fill=black] (-.1,1.9) circle (.02cm) node[above,coord] {$(1,1)$};
        \draw[black] (-.1,1.5) -- (-.1,1.9);
    \end{scope}
    \begin{scope}[shift={(5.25,0.3)},scale=0.95,transform shape]
        \draw[black!20,->] (-.5,1.5) -- (.5,1.5);
        \draw[black!20,->] (-.5,1.5) -- (-.5,2.5);
        \draw[fill=black] (-.1,1.5) circle (.02cm) node[below,coord] {$(1,0)$};
        \draw[fill=black] (-.5,1.9) circle (.02cm) node[left,coord] {$(0,1)$};
        \draw[fill=black] (-.1,1.9) circle (.02cm) node[above,coord] {$(1,1)$};
        \draw[black, fill=gray!60] (-.1,1.5) -- (-.1,1.9) -- (-.5,1.9) -- cycle;
    \end{scope}
    \begin{scope}[shift={(7,0.3)},scale=0.95,transform shape]
        \draw[black!20,->] (-.5,1.5) -- (.5,1.5);
        \draw[black!20,->] (-.5,1.5) -- (-.5,2.5);
        \draw[fill=black] (.3,1.9) circle (.02cm) node[right,coord] {$(2,1)$};
        \draw[fill=black] (-.1,2.3) circle (.02cm) node[left,coord] {$(1,2)$};
        \draw[fill=black] (.3,2.3) circle (.02cm) node[above,coord] {$(2,2)$};
        \draw[black, fill=gray!60] (.3,1.9) -- (.3,2.3) -- (-.1,2.3) -- cycle;
    \end{scope}
    \begin{scope}[shift={(8.75,0.3)},scale=0.95,transform shape]
        \draw[black!20,->] (-.5,1.5) -- (.5,1.5);
        \draw[black!20,->] (-.5,1.5) -- (-.5,2.5);
        \draw[fill=black] (-.1,1.5) circle (.02cm) node[below,coord] {$(1,0)$};
        \draw[fill=black] (.3,1.9) circle (.02cm) node[right,coord] {$(2,1)$};
        \draw[fill=black] (-.1,2.3) circle (.02cm) node[left,coord] {$(1,2)$};
        \draw[fill=black] (.3,2.3) circle (.02cm) node[above,coord] {$(2,2)$};
        \draw[black, fill=gray!60] (.3,1.9) -- (.3,2.3) -- (-.1,2.3) -- cycle;
        \draw[black, fill=gray!40] (.3,1.9) -- (-.1,1.5) -- (-.1,2.3) -- cycle;
    \end{scope}
    \begin{scope}[shift={(11,0.3)},scale=0.95,transform shape]
        \draw[black!20,->] (-.5,1.5) -- (.5,1.5);
        \draw[black!20,->] (-.5,1.5) -- (-.5,2.5);
        \draw[fill=black] (-.1,1.5) circle (.02cm) node[below,coord] {$(1,0)$};
        \draw[fill=black] (-.5,1.9) circle (.02cm) node[left,coord] {$(0,1)$};
        \draw[fill=black] (.3,1.9) circle (.02cm) node[right,coord] {$(2,1)$};
        \draw[fill=black] (-.1,2.3) circle (.02cm) node[left,coord] {$(1,2)$};
        \draw[fill=black] (.3,2.3) circle (.02cm) node[above,coord] {$(2,2)$};
        \draw[black, fill=gray!60] (.3,1.9) -- (.3,2.3) -- (-.1,2.3) -- cycle;
        \draw[black, fill=gray!40] (.3,1.9) -- (-.1,1.5) -- (-.1,2.3) -- cycle;
        \draw[black, fill=gray!20] (-.5,1.9) -- (-.1,1.5) -- (-.1,2.3) -- cycle;
    \end{scope}
    \begin{scope}[shift={(13.25,0.3)},scale=0.95,transform shape]
        \draw[black!20,->] (-.5,1.5) -- (1.1,1.5);
        \draw[black!20,->] (-.5,1.5) -- (-.5,3.1);
        \draw[fill=black] (.2,1.5) circle (.02cm) node[below,coord] {$(2,0)$};
        \draw[fill=black] (-.5,2.2) circle (.02cm) node[left,coord] {$(0,2)$};
        \draw[fill=black] (.9,2.2) circle (.02cm) node[right,coord] {$(4,2)$};
        \draw[fill=black] (.2,2.9) circle (.02cm) node[left,coord] {$(2,4)$};
        \draw[fill=black] (.9,2.9) circle (.02cm) node[above,coord] {$(4,4)$};
        \draw[black, fill=gray!60] (.9,2.2) -- (.9,2.9) -- (.2,2.9) -- cycle;
        \draw[black, fill=gray!40] (.9,2.2) -- (.2,1.5) -- (.2,2.9) -- cycle;
        \draw[black, fill=gray!20] (-.5,2.2) -- (.2,1.5) -- (.2,2.9) -- cycle;
    \end{scope}
    \end{tikzpicture}}
    \caption{
    Polytopal interpretation of an STC.
    Below each gate $v$ is the corresponding polytope $P_v$.
    }
    \label{fig:polytopal_viewpoint}
\end{figure}
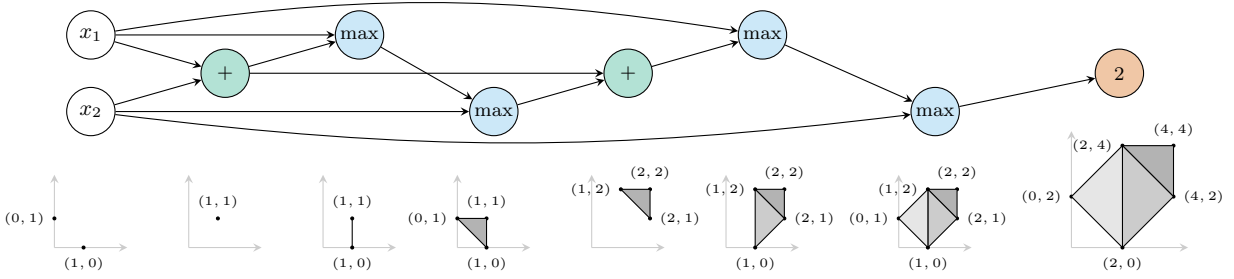

\section{A General Lower Bound Strategy}
In this section, we will develop the necessary prerequisites that allow us to reduce the problem of finding lower bounds to a more combinatorial / polytopal problem.
We will use the following two lemmas to transfer lower bounds for STCs to lower bounds for maxout ICNNs.
\begin{lemma}[\text{\citep[Proposition 2.3]{hertrich2023towards}}]\label{lem:bias_to_nonbias_for_pos_hom_functions}
    If a rank-$k$ maxout network computes a positively homogeneous function $f$, then the same network without biases also computes $f$.
\end{lemma}
\begin{lemma}[\text{\citep[Proposition 3.2]{hertrich2024neural}}]\label{lem:ICNN_to_mon_for_mon_functions}
    If a rank-$k$ maxout ICNN computes a monotone function $f$, then there is a monotone maxout network of the same size that computes $f$.
\end{lemma}
For polytopes $P\in\mathcal{P}^d_+$ the support function $f_P$ is monotone and positively homogeneous.
Therefore, a lower bound for bias-free monotone maxout networks implies the same lower bound for maxout ICNNs computing $f_P$.
It remains to connect monotone maxout networks to STCs.
\begin{lemma}\label{lem:mon_maxout_to_STC}
    If a bias-free monotone rank-$k$ maxout network of size $s$ computes a function $f:\R^d\to \R$, then there is an STC $\Phi$ with $\size(\Phi)\leq ks^2+ksd$ that computes the same function $f$.
\end{lemma}
\begin{proof}
    Fix a bias-free monotone rank-$k$ maxout network of size $s\geq 1$.
    Each maxout neuron computes an expression $z_v(x)=\max_{i=1,\dots, k}\sum_{u\in \delta^{-}_v}w^i_{uv}z_u(x)$ with weights $w^i_{uv}\geq 0$.
    For each of the $k$ affine functions, we use scalar gates for the nonzero weights and $|\delta^{-}_v|-1$ sum gates to add the weighted inputs.
    We then use $k-1$ $\max$ gates to take the maximum of the $k$ sums; see \Cref{fig:maxout_rank3}.
    If some weight is zero, we omit the corresponding scalar gate.
    If all weights of an affine function are zero, it is represented by the constant-0 input gate.
    Since $|\delta^{-}_v|\leq d+s$ for every maxout neuron $v$, one maxout neuron contributes at most $k(d+s-1)$ sum gates and $k-1$ $\max$ gates.
    Thus,
    \[
    \size(\Phi)\leq s(k(d+s-1)+(k-1))\leq ks(d+s)=ks^2+ksd.\qedhere
    \]
\end{proof}
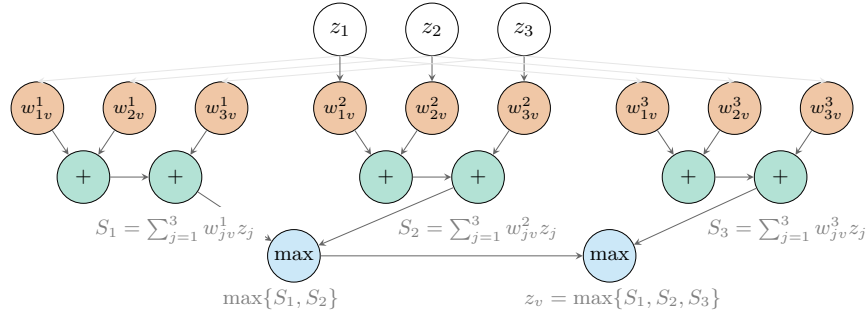
\begin{figure}
\centering
\resizebox{.7\linewidth}{!}{
\begin{tikzpicture}[
    >=stealth,
    gate/.style={
        draw,
        circle,
        minimum size=8mm,
        inner sep=0pt,
        align=center
    },
    input/.style={gate, font=\small},
    scalar/.style={gate, fill=OIorange!35, font=\scriptsize},
    sum/.style={gate, fill=OIgreen!30, font=\footnotesize},
    maxgate/.style={gate, fill=OIblue!30, font=\footnotesize},
    edge/.style={->, black!60},
    fadedge/.style={->, black!40},
    expr/.style={font=\footnotesize, align=center, black!70}
]

\node[input] (z1) at (7.8, 1) {$z_1$};
\node[input] (z2) at (9.2, 1) {$z_2$};
\node[input] (z3) at (10.6, 1) {$z_3$};

\node[scalar] (m11) at (3.2, -.2) {$w^1_{1v}$};
\node[scalar] (m12) at (4.6, -.2) {$w^1_{2v}$};
\node[scalar] (m13) at (6, -.2) {$w^1_{3v}$};

\node[sum] (p11) at (3.9, -1.25) {$+$};
\node[sum] (p12) at (5.3, -1.25) {$+$};

\node[scalar] (m21) at (7.8, -.2) {$w^2_{1v}$};
\node[scalar] (m22) at (9.2, -.2) {$w^2_{2v}$};
\node[scalar] (m23) at (10.6, -.2) {$w^2_{3v}$};

\node[sum] (p21) at (8.5, -1.25) {$+$};
\node[sum] (p22) at (9.9, -1.25) {$+$};

\node[scalar] (m31) at (12.4, -.2) {$w^3_{1v}$};
\node[scalar] (m32) at (13.8, -.2) {$w^3_{2v}$};
\node[scalar] (m33) at (15.2, -.2) {$w^3_{3v}$};

\node[sum] (p31) at (13.1, -1.25) {$+$};
\node[sum] (p32) at (14.5, -1.25) {$+$};

\node[maxgate] (max12)  at (7.1, -2.45) {$\max$};
\node[maxgate] (max123) at (11.9, -2.45) {$\max$};

\draw[fadedge] (z1.south) -- (m11.north);
\draw[fadedge] (z2.south) -- (m12.north);
\draw[fadedge] (z3.south) -- (m13.north);
\draw[edge] (z1.south) -- (m21.north);
\draw[edge] (z2.south) -- (m22.north);
\draw[edge] (z3.south) -- (m23.north);
\draw[fadedge] (z1.south) -- (m31.north);
\draw[fadedge] (z2.south) -- (m32.north);
\draw[fadedge] (z3.south) -- (m33.north);

\draw[edge] (m11) -- (p11);
\draw[edge] (m12) -- (p11);
\draw[edge] (p11) -- (p12);
\draw[edge] (m13) -- (p12);

\draw[edge] (m21) -- (p21);
\draw[edge] (m22) -- (p21);
\draw[edge] (p21) -- (p22);
\draw[edge] (m23) -- (p22);

\draw[edge] (m31) -- (p31);
\draw[edge] (m32) -- (p31);
\draw[edge] (p31) -- (p32);
\draw[edge] (m33) -- (p32);

\draw[edge] (p12) -- (max12);
\draw[edge] (p22) -- (max12);

\draw[edge] (max12) -- (max123);
\draw[edge] (p32) -- (max123);

\node[expr,fill=white] at (5.3, -2.05) {$S_1=\sum_{j=1}^3 w^1_{jv}z_j$};
\node[expr] at (9.9, -2.05) {$S_2=\sum_{j=1}^3 w^2_{jv}z_j$};
\node[expr] at (14.6, -2.05) {$S_3=\sum_{j=1}^3 w^3_{jv}z_j$};
\node[expr] at (6.9, -3.1) {$\max\{S_1,S_2\}$};
\node[expr] at (12.1, -3.1) {$z_v=\max\{S_1,S_2,S_3\}$};

\end{tikzpicture}}
\caption{
A rank-3 maxout node with three inputs represented as an STC.
Orange nodes are scalar gates, green nodes are sum gates, and blue nodes are $\max$ gates.}
\label{fig:maxout_rank3}
\end{figure}
\paragraph{Traces of STCs.}
For tropical circuit lower bounds it is useful to look at the monomials generated by the circuit.
In arithmetic circuit theory subcircuits that compute monomials are called parse trees \citep{jerrum1982some} or traces \citep{jukna2016tropical}.
We use the term trace.
A \emph{trace} $T$ of an STC $\Phi$ is a subgraph of the directed acyclic graph underlying $\Phi$ and is defined recursively from the output gate as follows.
Every trace $T$ contains the output gate.
Let $v$ be a gate that is already in $T$.
If $v$ is a $\max$ gate, then exactly one predecessor gate and the corresponding arc are included in $T$.
If $v$ is a $+$ gate, then both predecessor gates and the corresponding two arcs are included in $T$; if the two arcs have the same tail, both arcs are still part of the trace.
If $v$ is a scalar gate, then its predecessor gate and the corresponding arc are included in $T$.
See \Cref{fig:trace_example} and \Cref{fig:trace_proof4} for examples.

Each trace $T$ has a monomial coefficient vector $a_T\in \R^d_{\geq 0}$.
It is computed by the STC given by the trace $T$ where every $\max$ gate computes the function of its unique predecessor.
Equivalently, an input $x_i$ contributes $e_i$, the constant input $0$ contributes the zero vector, a scalar gate with scalar $\lambda$ multiplies the vector by $\lambda$, and a $+$ gate adds the vectors coming from its two incoming arcs.

\begin{lemma}\label{lem:trace_union}
    Let $\mathcal{T}$ be the set of all traces of an STC $\Phi$, and let $a_T$ be the monomial coefficient vector of a trace $T$.
    Then
    \[
    f_\Phi(x)=\max_{T\in \mathcal{T}}a_T^\top x,
    \qquad
    P_\Phi = \conv\{a_T: T\in \mathcal{T}\}.
    \]
\end{lemma}
\begin{proof}
    Every trace $T$ corresponds to an affine function $a_T^\top x$ with $a_T^\top x\leq f_\Phi(x)$ for all $x\in \R^d$, since each max gate of $\Phi$ is replaced by only one of its inputs.
    Hence $\max_{T\in \mathcal{T}}a_T^\top x\leq f_\Phi(x)$.
    Fix $x\in \R^d$.
    From the output gate to the inputs, choose at every $\max$ gate a predecessor gate with the maximum value at $x$.
    This induces a trace $T_x$ with $f_\Phi(x)=a_{T_x}^\top x\leq \max_{T\in \mathcal{T}}a_T^\top x$ and proves equality.
\end{proof}

We next prove the decomposition lemma that we use in our lower-bound arguments.
It is a tropical and polytopal analogue of the fact that a polynomial computed by a small arithmetic circuit can be written as a sum of a small number of products of ``simpler'' polynomials.
This general technique has been widely used, for example, in \citep{hyafil1978parallel,valiant1980negation,hrubevs2011homogeneous,raz2011multilinear,jukna2015lower,jukna2019greedy,yehudayoff2019separating,srinivasan2020strongly,chattopadhyay2022monotone,kluk2025lower,cavalar2026negations}, also compare the discussions in \citep{shpilka2010arithmetic,jukna2016tropical}.
\paragraph{Measure.}
We call a function $\mu:\mathcal{P}_+^d\to \R_{\geq 0}$ a \emph{measure} if $\mu(\{0\})\leq 1$, $\mu(\{e_i\}) \leq 1$ for all $i\in [d]$, and, for all $A,B\in \mathcal{P}_+^d$ and all $\lambda\in \R_{>0}$,
\[
\mu(\lambda  A) = \mu(A), \qquad\mu(A+B) \leq \mu(A) + \mu(B), \quad\text{and}\quad \mu(\conv(A \cup B)) \leq \mu(A) + \mu(B).
\]
We use the convention $\mu(\emptyset)=0$.
Later, in \Cref{sec:birkhoff_polytope,sec:dst_polytope}, we will use measures of the form $P\mapsto \max_{v\in V(P)}|\supp(\pi(v))|$, where $\pi:\R^d\to \R^m$ is a linear map.

\begin{lemma}[Decomposition Lemma]\label{lem:decomposition_lower_bound}
Let $r \geq 1$, let $\mu:\mathcal{P}_+^d\to \R_{\geq 0}$ be a measure, and let $\Phi$ be an STC with $\size(\Phi)=t$ computing the support function of a polytope $Q\in \mathcal{P}_+^d$.
Then there are an integer $k \leq t$ and polytopes $A_0, A_1, B_1, \dots, A_k, B_k \in \mathcal{P}^d_+$ such that
\begin{equation*}
    Q = \mathrm{conv}\left(A_0 \cup \bigcup_{i=1}^k (A_i + B_i)\right)
\end{equation*}
where $\mu(A_0) \leq r$ and $r < \mu(A_i) \leq 2r$ for all $i = 1, \dots, k$.
\end{lemma}
\begin{proof}
    We use induction on $t$.
    For $t=0$, the output is obtained from an input gate using only scalar gates.
    Hence $Q$ is either $\{0\}$ or a dilation of some $\{e_i\}$.
    In both cases $\mu(Q)\leq 1\leq r$, so the statement holds with $k=0$ and $A_0=Q$.
    
    Assume now $t\geq 1$.
    If $\mu(Q)\leq r$, then the statement again holds with $k=0$ and $A_0=Q$.
    Thus, assume $\mu(Q)>r$.
    Starting at the output gate, we move backwards to the input gates by recursively choosing a predecessor gate of maximum measure as long as some predecessor gate has measure greater than $r$.
    Since all input gates have measure at most $1\leq r$, this process stops at a $\max$ or $+$ gate $v$ with $\mu(P_v)>r$ with predecessor gates $u,w$ such that $\mu(P_u)\leq r$ and $\mu(P_w)\leq r$.
    Since $\mu$ is a measure, $r<\mu(P_v)\leq \mu(P_u)+\mu(P_w)\leq 2r$.
    Note that $v$ cannot be a scalar gate, as scalar gates do not change the measure.

    Let $\Phi_v$ be the circuit obtained from $\Phi$ by deleting the arcs entering $v$ and treating $v$ as an additional input gate carrying the variable $x_{d+1}$.
    Non-input gates and arcs that cannot reach the output gate after the removal of the arcs are deleted; see \Cref{fig:trace_proof2} for an example.
    For every $x\in \R^d$, we have
    \[
    f_\Phi(x)=f_{\Phi_v}(x,f_v(x)),
    \]
    since after substituting $x_{d+1}=f_v(x)$, the value at $v$ in $\Phi_v$ equals the value at $v$ in $\Phi$ and all other gates in $\Phi_v$ compute the same values as in $\Phi$.

    Let $\mathcal{T}$ be the set of traces of $\Phi_v$.
    The monomial coefficient vector of a trace $T\in \mathcal{T}$ is $(a_T,c_T)\in \R^d_{\geq 0}\times \R_{\geq 0}$ and the trace $T$ computes the monomial $a_T^\top x+c_T x_{d+1}$.  
    By \Cref{lem:trace_union}, we have
    \[
    f_{\Phi_v}(x,x_{d+1})=\max_{T\in \mathcal{T}}(a_T^\top x+c_T x_{d+1}).
    \]
    Using $f_v=f_{P_v}$ and $c_T\geq 0$ it follows that
    \begin{align*}
    f_Q(x)
    =f_\Phi(x)
    =f_{\Phi_v}(x,f_v(x))
    =\max_{T\in \mathcal{T}}\left(a_T^\top x+c_T f_v(x)\right)
    =\max_{c\in \conv\{a_T+c_T p:\; T\in \mathcal{T},\,p\in P_v\}}c^\top x.
    \end{align*}
    Thus
    \[
    Q=\conv\{a_T+c_T p:\; T\in \mathcal{T},\,p\in P_v\}=\conv\left(\bigcup_{T\in \mathcal{T}}\left(\{a_T\}+c_TP_v\right)\right).
    \]

    We now partition $\mathcal{T}$ into two parts.
    Let $\mathcal{T}_v$ be the set of traces that contain gate $v$, and let $\mathcal{T}_{-v}\coloneq \mathcal{T}\setminus \mathcal{T}_v$.
    If $T\in \mathcal{T}_{-v}$, then $c_T=0$.
    If $T\in \mathcal{T}_v$, then $c_T>0$, since the trace contains a directed path from $v$ to the output gate and along this path all scalar factors are positive, while $+$ gates only add other coefficients.
    We have $Q=\conv(R\cup S)$, where
    \[
    R\coloneq \conv\left(\bigcup_{T\in \mathcal{T}_v}\left(\{a_T\}+c_TP_v\right)\right)
    \]
    and $S\coloneq \conv\{a_T: T\in \mathcal{T}_{-v}\}$.
    With $c_{\min}=\min_{T\in \mathcal{T}_v}c_T>0$, we have
    \[
    R=
    \conv\left(\bigcup_{T\in \mathcal{T}_v}\left(\{a_T\}+c_TP_v\right)\right)
    =
    c_{\min}\cdot P_v+\conv\left(\bigcup_{T\in \mathcal{T}_v}\left(\{a_T\}+(c_T-c_{\min})\cdot P_v\right)\right)\eqcolon A^*+B^*,
    \]
    with $A^*,B^*\in \mathcal{P}^d_+$ and $\mu(A^*)=\mu(c_{\min}P_v)=\mu(P_v)\in (r,2r]$.

    It remains to consider $S$.
    If $\mathcal{T}_{-v}$ is empty, we keep $S=\emptyset$.
    Otherwise, let $\Phi_{-v}$ be the circuit obtained from $\Phi_v$ by deleting the input gate $v$ and all arcs leaving it, and then recursively modifying the remaining circuit as follows.
    Delete scalar gates with no input, delete $+$ gates with fewer than two inputs, and contract $\max$ gates with only one input by replacing the arcs leaving the $\max$ gate by direct arcs from its unique predecessor to its successor gates and deleting the $\max$ gate; see \Cref{fig:trace_proof3} for an illustration.

    These simplifications preserve exactly the traces of $\Phi_v$ that do not contain $v$, since a trace without $v$ is not affected by the deletions, and contracting $\max$ gates with only one input does not change the monomial coefficient vector of any trace.
    Therefore the traces of $\Phi_{-v}$ are in bijection with the traces in $\mathcal{T}_{-v}$, with the same monomial coefficient vectors.
    By \Cref{lem:trace_union}, $\Phi_{-v}$ computes the support function of $S$.
    Also, $\size(\Phi_{-v})\leq t-1$.

    If $S\neq\emptyset$, we apply the induction hypothesis to $\Phi_{-v}$.
    This gives an integer $\ell \leq t-1$ and polytopes $A_0,A_1,B_1,\dots, A_\ell,B_\ell\in \mathcal{P}^d_+$ with $S=\conv\left(A_0\cup \bigcup_{i=1}^\ell\left(A_i+B_i\right)\right)$ and $\mu(A_0)\leq r$ and $r<\mu(A_i)\leq 2r$ for all $i=1,\dots, \ell$.
    If $S=\emptyset$, we can use $\ell=0$ and $A_0=\emptyset$.
    Then, with $A_{\ell+1}\coloneq A^*$ and $B_{\ell+1}\coloneq B^*$, we obtain 
    \[
    Q=\conv\left(A_0\cup \bigcup_{i=1}^{\ell+1}\left(A_i+B_i\right)\right)
    \]
    with $\ell+1\leq t$.
    This completes the induction.
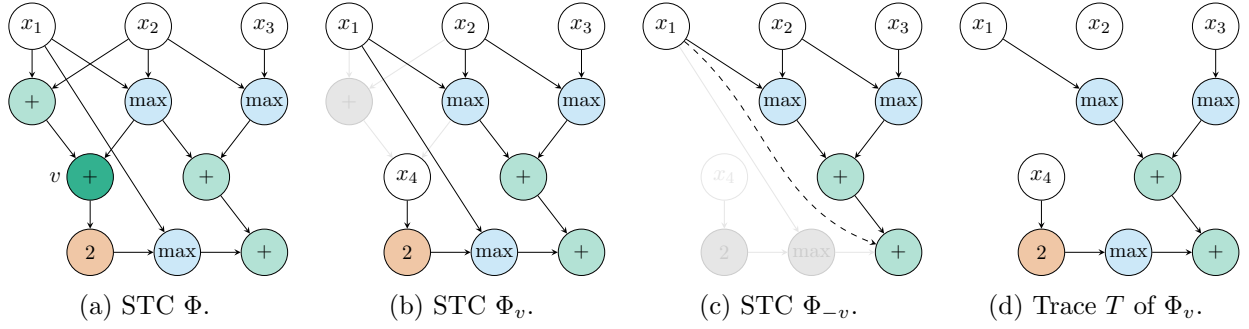
\begin{figure}
\centering
\begin{subfigure}{0.23\textwidth}
\centering
\resizebox{\linewidth}{!}{
\begin{tikzpicture}[>=stealth]
\tikzset{
    inp/.style={draw,circle,minimum size=8mm,inner sep=0pt,font=\normalsize},
    scalar/.style={draw,circle,minimum size=8mm,inner sep=0pt,font=\small, fill=OIorange!35},
    plus/.style={draw,circle,minimum size=8mm,inner sep=0pt,font=\small, fill=OIgreen!30},
    maxg/.style={draw,circle,minimum width=8mm,inner sep=0pt,font=\small, fill=OIblue!30},
    expr/.style={font=\scriptsize, align=center}
}

\node[inp] (x1) at (1,0) {$x_1$};
\node[inp] (x2) at (3,0) {$x_2$};
\node[inp] (x3) at (5,0) {$x_3$};

\node[plus] (a1) at (1,-1.3) {$+$};         
\node[maxg] (a2) at (3,-1.3) {$\max$};      
\node[maxg] (a3) at (5,-1.3) {$\max$};      
\node[plus,fill=OIgreen!80] (b1) at (2,-2.6) {$+$};         
\node[black] at (1.4,-2.6) {$v$};
\node[plus] (b2) at (4,-2.6) {$+$};         
\node[scalar] (c1) at (2,-3.9) {$2$};        
\node[maxg] (c2) at (3.5,-3.9) {$\max$};      
\node[plus] (d) at (5,-3.9) {$+$};          

\draw[->] (x1) -- (a1);
\draw[->] (x2) -- (a1);
\draw[->] (x1) -- (a2);
\draw[->] (x2) -- (a2);
\draw[->] (x2) -- (a3);
\draw[->] (x3) -- (a3);
\draw[->] (a1) -- (b1);
\draw[->] (a2) -- (b1);
\draw[->] (a2) -- (b2);
\draw[->] (a3) -- (b2);
\draw[->] (b1) -- (c1);
\draw[->] (c1) -- (c2);
\draw[->] (x1) -- (c2);
\draw[->] (c2) -- (d);
\draw[->] (b2) -- (d);
\end{tikzpicture}}
\caption{STC $\Phi$.}
\label{fig:trace_proof1}
\end{subfigure}
\hfill
\begin{subfigure}{0.23\textwidth}
\centering
\resizebox{\linewidth}{!}{
\begin{tikzpicture}[>=stealth]
\tikzset{
    inp/.style={draw,circle,minimum size=8mm,inner sep=0pt,font=\normalsize},
    scalar/.style={draw,circle,minimum size=8mm,inner sep=0pt,font=\small, fill=OIorange!35},
    plus/.style={draw,circle,minimum size=8mm,inner sep=0pt,font=\small, fill=OIgreen!30},
    maxg/.style={draw,circle,minimum width=8mm,inner sep=0pt,font=\small, fill=OIblue!30},
    expr/.style={font=\scriptsize, align=center}
}

\node[inp] (x1) at (1,0) {$x_1$};
\node[inp] (x2) at (3,0) {$x_2$};
\node[inp] (x3) at (5,0) {$x_3$};

\node[plus,black!40,fill=black!10] (a1) at (1,-1.3) {$+$};
\node[maxg] (a2) at (3,-1.3) {$\max$};      
\node[maxg] (a3) at (5,-1.3) {$\max$};      
\node[inp] (b1) at (2,-2.6) {$x_4$};
\node[plus] (b2) at (4,-2.6) {$+$};         
\node[scalar] (c1) at (2,-3.9) {$2$};        
\node[maxg] (c2) at (3.5,-3.9) {$\max$};      
\node[plus] (d) at (5,-3.9) {$+$};          

\draw[->,black!40] (x1) -- (a1);
\draw[->,black!40] (x2) -- (a1);
\draw[->] (x1) -- (a2);
\draw[->] (x2) -- (a2);
\draw[->] (x2) -- (a3);
\draw[->] (x3) -- (a3);
\draw[->,black!40] (a1) -- (b1);
\draw[->,black!40] (a2) -- (b1);
\draw[->] (a2) -- (b2);
\draw[->] (a3) -- (b2);
\draw[->] (b1) -- (c1);
\draw[->] (c1) -- (c2);
\draw[->] (x1) -- (c2);
\draw[->] (c2) -- (d);
\draw[->] (b2) -- (d);
\end{tikzpicture}}
\caption{STC $\Phi_v$.}
\label{fig:trace_proof2}
\end{subfigure}
\hfill
\begin{subfigure}{0.23\textwidth}
\centering
\resizebox{\linewidth}{!}{
\begin{tikzpicture}[>=stealth]
\tikzset{
    inp/.style={draw,circle,minimum size=8mm,inner sep=0pt,font=\normalsize},
    scalar/.style={draw,circle,minimum size=8mm,inner sep=0pt,font=\small, fill=OIorange!35},
    plus/.style={draw,circle,minimum size=8mm,inner sep=0pt,font=\small, fill=OIgreen!30},
    maxg/.style={draw,circle,minimum width=8mm,inner sep=0pt,font=\small, fill=OIblue!30},
    expr/.style={font=\scriptsize, align=center}
}
\node[inp] (x1) at (1,0) {$x_1$};
\node[inp] (x2) at (3,0) {$x_2$};
\node[inp] (x3) at (5,0) {$x_3$};

\node[maxg] (a2) at (3,-1.3) {$\max$};      
\node[maxg] (a3) at (5,-1.3) {$\max$};      
\node[inp,black!40] (b1) at (2,-2.6) {$x_4$};
\node[plus] (b2) at (4,-2.6) {$+$};         
\node[scalar,black!40,fill=black!10] (c1) at (2,-3.9) {$2$};
\node[maxg,black!40,fill=black!10] (c2) at (3.5,-3.9) {$\max$};
\node[plus] (d) at (5,-3.9) {$+$};          

\draw[->] (x1) -- (a2);
\draw[->] (x2) -- (a2);
\draw[->] (x2) -- (a3);
\draw[->] (x3) -- (a3);
\draw[->] (a2) -- (b2);
\draw[->] (a3) -- (b2);
\draw[->,black!40] (b1) -- (c1);
\draw[->,black!40] (c1) -- (c2);
\draw[->,black!40] (x1) -- (c2);
\draw[->,dashed] (x1) to[out=-40,in=160] (d);
\draw[->,black!40] (c2) -- (d);
\draw[->] (b2) -- (d);
\end{tikzpicture}}
\caption{STC $\Phi_{-v}$.}
\label{fig:trace_proof3}
\end{subfigure}
\hfill
\begin{subfigure}{0.23\textwidth}
\centering
\resizebox{\linewidth}{!}{
\begin{tikzpicture}[>=stealth]
\tikzset{
    inp/.style={draw,circle,minimum size=8mm,inner sep=0pt,font=\normalsize},
    scalar/.style={draw,circle,minimum size=8mm,inner sep=0pt,font=\small, fill=OIorange!35},
    plus/.style={draw,circle,minimum size=8mm,inner sep=0pt,font=\small, fill=OIgreen!30},
    maxg/.style={draw,circle,minimum width=8mm,inner sep=0pt,font=\small, fill=OIblue!30},
    expr/.style={font=\scriptsize, align=center}
}

\node[inp] (x1) at (1,0) {$x_1$};
\node[inp] (x2) at (3,0) {$x_2$};
\node[inp] (x3) at (5,0) {$x_3$};

\node[maxg] (a2) at (3,-1.3) {$\max$};      
\node[maxg] (a3) at (5,-1.3) {$\max$};      
\node[inp] (b1) at (2,-2.6) {$x_4$};
\node[plus] (b2) at (4,-2.6) {$+$};         
\node[scalar] (c1) at (2,-3.9) {$2$};        
\node[maxg] (c2) at (3.5,-3.9) {$\max$};      
\node[plus] (d) at (5,-3.9) {$+$};          

\draw[->] (x1) -- (a2);
\draw[->] (x3) -- (a3);
\draw[->] (a2) -- (b2);
\draw[->] (a3) -- (b2);
\draw[->] (b1) -- (c1);
\draw[->] (c1) -- (c2);
\draw[->] (c2) -- (d);
\draw[->] (b2) -- (d);
\end{tikzpicture}}
\caption{Trace $T$ of $\Phi_v$.}
\label{fig:trace_proof4}
\end{subfigure}
\label{fig:trace_proof}
\caption{
Illustration of circuits in the proof of \Cref{lem:decomposition_lower_bound}.
The STC $\Phi$ in (a) computes $f_Q$ for $Q=\conv\{(5,3,0),(5,2,1),(2,6,0),(2,5,1),(2,1,0),(1,2,0),(2,0,1),(1,1,1)\}$.
The gate $v$ in (a) computes $f_{P_v}$ for $P_v=\conv\{(2,1,0),(1,2,0)\}$.
The STC $\Phi_v$ in (b) computes $f_P$ for $P=\conv\{(1,1,0,2),(1,0,1,2),(0,2,0,2),(0,1,1,2),(2,1,0,0),(1,2,0,0),(2,0,1,0),(1,1,1,0)\}$.
The STC $\Phi_{-v}$ in (c) computes $f_S$ for $S=\{(2,1,0),(1,2,0),(2,0,1),(1,1,1)\}$.
The monomial coefficient vector of the trace $T$ in (d) is $(1,0,1,2)$.
In (b) and (c), greyed out gates and arcs are deleted; the new dashed arc comes from the deletion of the $\max$ gate.
}
\label{fig:rectangle_proof}
\end{figure}
\end{proof}
\Cref{lem:decomposition_lower_bound} can be used to prove lower bounds on the size of STCs.
The idea is to choose a measure and then show that every decomposition of the form in the lemma requires many summands $A_i+B_i$.
Additionally, if $\max$ gates do not increase the measure, then the same argument gives a lower bound on the number of $+$ gates.
\begin{corollary}\label{cor:general_bound}
    Let $r \geq 1$, let $Q\in \mathcal{P}^d_+$, and let $\mu:\mathcal{P}_+^d\mapsto \R_{\geq 0}$ be a measure.
    Suppose that for every collection $A_0, A_1, B_1, \dots, A_k, B_k \in \mathcal{P}^d_+$ with $\mu(A_0) \leq r$, $r < \mu(A_i) \leq 2r$ for all $i = 1, \dots, k$, and $Q = \mathrm{conv}\left(A_0 \cup \bigcup_{i=1}^k (A_i + B_i)\right)$, it holds that $L\leq k$.
    Then every STC computing $f_Q$ satisfies $L\leq \size(\Phi)$.
    Moreover, if $\mu(\conv(A \cup B)) \leq \max(\mu(A),\mu(B))$ for all $A,B\in \mathcal{P}^d_+$, then every STC computing $f_Q$ satisfies $L\leq \size_+(\Phi)$.
\end{corollary}
\begin{proof}
    The lower bound $L\leq \size(\Phi)$ follows directly from \Cref{lem:decomposition_lower_bound}.
    For the stronger statement, assume that $\mu(\conv(A \cup B)) \leq \max(\mu(A),\mu(B))$ for all $A,B\in \mathcal{P}^d_+$.
    In the proof of \Cref{lem:decomposition_lower_bound}, the selected gate $v$ and its predecessor gates $u,w$ satisfy $\mu(P_v)>r$, $\mu(P_u)\leq r$, and $\mu(P_w)\leq r$.
    Such a gate cannot be a $\max$ gate under the additional assumption and must therefore be a $+$ gate.
    Thus each summand $A_i+B_i$ in the decomposition comes from a distinct $+$ gate, and the same argument gives $L\leq \size_+(\Phi)$.
\end{proof}

For homogeneous multilinear polynomials $f$ with many monomials, variants of \Cref{cor:general_bound} using the degree as a measure have been (implicitly) used to prove lower bounds for arithmetic and tropical circuits \citep{shpilka2010arithmetic,jukna2015lower,jukna2019greedy}.
In our setting, this degree measure corresponds to $A\mapsto \max_{v\in V(A)}|\supp(v)|$.
A common strategy is to show that each summand $A_i+B_i$ contains only few monomials, which then yields a lower bound on the number of summands $A_i+B_i$ necessary to produce \emph{all} monomials of $f$.
Our proofs follow existing counting ideas from the lower bound proofs for arithmetic and tropical circuits in \citep{shpilka2010arithmetic,jukna2016tropical,jukna2023tropical}, but we have to account for fractional vertices and use different measures.

\section{Lower Bounds for the Birkhoff Polytope}
\label{sec:birkhoff_polytope}

\paragraph{Rectangles.}
A \emph{rectangle} of a polytope $P\in \mathcal{P}^d_+$ is a pair of polytopes $(X,Y)$ with $X,Y \in \mathcal{P}^d_+$ and $X+Y\subseteq P$.
A rectangle is \emph{nonempty} if $X+Y\neq \emptyset$.
It is \emph{vertex-realizing} if $(X+Y)\cap V(P)\neq \emptyset$.

Let $M_n \subset \{0,1\}^{n \times n}$ be the set of characteristic vectors of perfect matchings of the complete bipartite graph $K_{n,n}$.
Recall that a matching is a set of edges where no two edges share a common node.
A matching $M$ is perfect if every node is incident to exactly one edge in $M$.
The \emph{Birkhoff Polytope} is $P_{\rm PERM} = \text{conv}(M_n)$.
It has the standard halfspace representation
\[
P_{\rm PERM}=\left\{z\in \R_{\geq 0}^{n \times n}: \sum_{j=1}^n z_{ij}=1,\, i\in [n],\, \sum_{i=1}^n z_{ij}=1,\, j\in [n]\right\}.
\]
Our goal is now to bound the number of perfect matchings in a single rectangle $X+Y\subseteq P_{\rm PERM}$ with $n/3<\mu(X)\leq 2n/3$ for a suitable measure $\mu$.

For $z \in \mathbb{R}^{n \times n}$ we define the node weights $a_i(z) \coloneq \sum_{j=1}^n z_{ij}$ and $b_j(z) \coloneq \sum_{i=1}^n z_{ij}$.
\begin{lemma}
\label{lem:matching:fixed_node_weights}
Let $(X,Y)$ be a nonempty rectangle of $P_{\rm PERM}$.
Then there are vectors $a^*, b^* \in [0,1]^n$ such that, for all $x \in X$,
\[
a_i(x) = a^*_i \quad \text{and} \quad b_j(x) = b^*_j, \quad  i,j \in [n].
\]
Consequently, for all $y \in Y$ and $i,j\in[n]$, we have $a_i(y)=1 - a^*_i$ and $b_j(y)=1- b^*_j$.
\end{lemma}
\begin{proof}
Fix an arbitrary $y \in Y$.
Then, since $X+Y\subseteq P_{\rm PERM}$, we have $x + y \in P_{\rm PERM}$ and thus $a_i(x)=1-a_i(y)\eqcolon a_i^*$ for all $x\in X$ and $i\in [n]$.
Analogously, we obtain the statements for the node weights $b_j$ and for $y\in Y$.
\end{proof}

The next lemma shows that a rectangle induces a certain structure for the perfect matchings it contains.

\begin{lemma}
\label{lem:restricted_matchings}
Let $(X,Y)$ be a vertex-realizing rectangle of $P_{\rm PERM}$, let $a^*$ and $b^*$ be as in \Cref{lem:matching:fixed_node_weights}, and define
\[
I \coloneq \{i\in [n]: a^*_i > 0\}, \qquad J \coloneq \{j\in [n]: b^*_j > 0\}.
\]
Then, $|I|=|J|$, and every $m \in (X+Y) \cap M_n$ satisfies $\supp(m) \subseteq (I \times J) \cup (I^c \times J^c)$.
\end{lemma}
\begin{proof}
Let $m \in (X+Y) \cap M_n$ and write $m = x+y$ with $x \in X,\, y \in Y$.
We have $0\leq x\leq m$ and $0\leq y \leq m$.
Since $m$ is the characteristic vector of a perfect matching and $x\leq m$ holds, the edges in $\supp(x)$ form a matching $m_x$ in $K_{n,n}$.
In particular, each node $i\in I$ is incident to an edge in $m_x$, and each node $i\notin I$ is not incident to an edge in $m_x$.
The same applies to the nodes $j\in J$ and $j\notin J$.
Therefore $|I|=|J|$ and $\supp(x)\subseteq I\times J$.
Now consider an edge $(i,j)\in \supp(m)$.
If $(i,j)\in \supp(x)$, then $(i,j)\in I\times J$.
If $(i,j)\notin \supp(x)$, then $(i,j)\in \supp(y)$, in which case $i\notin I$ and the matching edge $(i,j)$ must use a node $j\notin J$.
Hence $\supp(m) \subseteq (I \times J) \cup (I^c \times J^c)$.
\end{proof}
In particular, if $(X,Y)$ is vertex-realizing, that is, $(X+Y)\cap M_n\neq \emptyset$, then
\[
|(X+Y)\cap M_n|\leq |I|!\cdot (n-|I|)!.
\]
We now choose a measure $\mu_1$ such that the restriction $\mu_1(X)\in (n/3,2n/3]$ forces $|I|$ to be bounded away from both $0$ and $n$.
Let
\[
\mu_1:\mathcal{P}_+^{n\times n}\to \R_{\geq 0},
\qquad
\textstyle A\mapsto\max_{a\in V(A)}|\{i\in [n]: \sum_{j=1}^n a_{ij}>0\}|.
\]
Again, for $A=\emptyset$, we use $\mu_1(\emptyset)=0$.
The measure $\mu_1$ corresponds to the degree-measure $A\mapsto\max_{v\in V(A)}|\supp(v)|$ after replacing $|\supp(v)|$ with $|\supp(\pi(v))|$ for a linear map $\pi:\R^{n\times n}\to \R^n$.

\begin{proposition}\label{prop:measure_matching}
    The function $\mu_1$ is a measure.
    Moreover, for all $A, B\in \mathcal{P}^{n\times n}_+$,
    \[
    \mu_1(\mathrm{conv}(A\cup B))\leq \max(\mu_1(A),\mu_1(B)).
    \]
\end{proposition}
\begin{proof}
    We have $\mu_1(\{0\})=0$ and $\mu_1(\{e_{ij}\})=1$ for all $i,j\in [n]$.
    It is straightforward to show that $\mu_1(c\cdot A)=\mu_1(A)$ and $\mu_1(A+B)\leq \mu_1(A)+\mu_1(B)$ for all $A, B\in \mathcal{P}^{n\times n}_+, c\in \R_{>0}$.
    Further, every vertex of $\conv(A\cup B)$ is a vertex of $A$ or a vertex of $B$.
    Thus $\mu_1(\conv(A\cup B))\leq\max(\mu_1(A),\mu_1(B))$.
\end{proof}

\begin{lemma}\label{lem:few_matchings_per_rectangle}
    Let $(X,Y)$ be a vertex-realizing rectangle of $P_{\rm PERM}$ with $\mu_1(X)=r\in (n/3,2n/3]$.
    Then
    \[
    |(X+Y)\cap M_n|\leq  \frac{n!}{\binom{n}{\lfloor n/3\rfloor}}.
    \]
\end{lemma}
\begin{proof}
    By \Cref{lem:matching:fixed_node_weights}, $\sum_{j=1}^n x_{ij}=a_i^*$ for all $x\in X$.
    With $I=\{i\in[n]: a_i^*>0\}$, we have
    \[
    r=\mu_1(X)=\max_{x\in V(X)}|\{i\in [n]: \sum_{j=1}^n x_{ij}>0\}|=\max_{x\in V(X)}|\{i\in[n]: a_i^*>0\}|=|I|.
    \]
    By \Cref{lem:restricted_matchings} and since $r$ ranges over $(n/3,2n/3]\cap \N$, we have
    \[
    |(X+Y)\cap M_n|\leq r!(n-r)!=\frac{n!}{\binom{n}{r}}\leq \frac{n!}{\binom{n}{\lfloor n/3\rfloor}}.
    \qedhere
    \]
\end{proof}
This allows us to state our lower bound for STCs computing $P_{\rm PERM}$.
\matchingstclowerbound*
\begin{proof}
    Consider a decomposition 
    \[
    P_{\rm PERM} = \mathrm{conv}\left(A_0 \cup \bigcup_{i=1}^k (A_i + B_i)\right)
    \]
    as in \Cref{cor:general_bound} with $r=n/3$, so $\mu_1(A_0) \leq n/3$, $n/3 < \mu_1(A_i) \leq 2n/3$ for all $i = 1, \dots, k$.
    Suppose that $A_0\neq \emptyset$.
    Then, $A_0$ has a vertex $v\in V(A_0)$ which lies in $P_{\rm PERM}$ (since $A_0\subseteq P_{\rm PERM}$) and satisfies $\mu_1(\{v\})\leq n/3$.
    However, this contradicts the fact that $\mu_1(\{x\})=n$ for all $x\in P_{\rm PERM}$.
    Thus we must have $A_0=\emptyset$.

    Each of the $n!$ perfect matchings must be contained in at least one rectangle.
    By \Cref{lem:few_matchings_per_rectangle} each rectangle contains at most $\frac{n!}{\binom{n}{\lfloor n/3\rfloor}}$ perfect matchings.
    Thus
    \[
    k\geq \binom{n}{\lfloor n/3\rfloor}\in 2^{\Omega(n)}.
    \]
    The lower bound on $\size_+(\Phi)$ follows from \Cref{cor:general_bound} and \Cref{prop:measure_matching}.
\end{proof}
We obtain the following lower bound on $\mnnc(P_{\rm PERM})$.
\matchingmnnc*
\begin{proof}
    Let $s$ be the size of a rank-2 maxout ICNN computing $f_{P_{\rm PERM}}$.
    Since $f_{P_{\rm PERM}}$ is positively homogeneous and monotone, \Cref{lem:bias_to_nonbias_for_pos_hom_functions} and \Cref{lem:ICNN_to_mon_for_mon_functions} imply that there is a bias-free monotone rank-2 maxout network of size $s$ computing $f_{P_{\rm PERM}}$.
    By \Cref{lem:mon_maxout_to_STC} this gives an STC $\Phi$ with $\size(\Phi)\leq 2s^2+2sn^2$.
    By \Cref{thm:matching_lower_bound}, $\size(\Phi)\geq 2^{\Omega(n)}$.
    Thus $2s^2+2sn^2\geq 2^{\Omega(n)}$, which implies $s\geq 2^{\Omega(n)}$.
\end{proof}
Although one can optimize over $P_{\rm PERM}$ in polynomial time, it is open whether there are $(\max,+,-)$-circuits or ReLU / maxout networks of polynomial size computing $f_{P_{\rm PERM}}$; see \citep[Section 6.5, Problem 3]{jukna2023tropical}.
Thus, the lower bound above does not by itself separate $\mnnc$ from $\nnc$.
It does show, however, that the extension-complexity lower bound $\xc(P)/2\leq \mnnc(P)$ from \citep{hertrich2024neural} can be exponentially loose, because $\xc(P_{\rm PERM})=n^2$ for $n\geq 4$ \citep[Proposition 5.10]{fiorini2013combinatorial}.

Dropping polynomial factors, the $2^{\Omega(n)}$ lower bound from \Cref{thm:matching_lower_bound} is tight up to a multiplicative factor in the exponent.
This was already observed by \citet{jerrum1982some}.
\begin{proposition}[\text{\citep[Section 4.3]{jerrum1982some}}]
    There is a tropical circuit of size $\mathcal{O}(n2^n)$ computing $f_{P_{\rm PERM}}$.
\end{proposition}
\begin{proof}
    \citet{jerrum1982some} describe a $(+,\times)$-circuit of size $\mathcal{O}(n2^n)$ for computing the permanent.
    Tropicalizing this circuit, that is, replacing $+$ gates by $\max$ gates and $\times$ gates by $+$ gates, leads to the following dynamic program, which we sketch here for completeness.
    For $I\subseteq [n]$ with $|I|=|J|$, let $M(I)$ be the maximum weight of a matching that matches the first $|I|$ nodes on the left to the nodes on the right in $I$.
    We have $M(\emptyset)=0$.
    For $I\neq \emptyset$, we have with $\ell=|I|$
    \[
    M(I)=\max_{i\in I}\{M(I\setminus \{i\})+x_{\ell i}\}.
    \]
    The output is $M([n])$.
    We show that the recurrence is correct.
    Let $M$ be a matching on the subgraph induced by $[\ell]\times I$ of maximum weight $W$ and let node $\ell$ on the left be matched to node $v\in I$ on the right.
    Since $W=(W-x_{\ell v})+x_{\ell v}=M(I\setminus \{v\})+x_{\ell v}$, we have $W\leq M(I)$.
    Moreover, each choice $i\in I$ in the recursion corresponds to exactly one matching on $[\ell]\times I$, since $M(I\setminus \{i\})$ corresponds to a matching $M'$ on $[\ell-1]\times I\setminus \{i\}$ and $M'\cup \{(\ell,i)\}$ is a matching on $[\ell]\times I$.
    Thus $W\geq M(I)$, which proves equality and shows that the recurrence is correct.
    There are $2^n$ subsets $I\subseteq [n]$ and for each subset $I$, a total of $|I|$ $+$ operations and $|I|-1$ $\max$ operations are performed.
    This gives a tropical circuit of size $\mathcal{O}(n2^n)$.
\end{proof}

\section{Lower Bounds for the Directed Spanning Tree Polytope}
\label{sec:dst_polytope}
Let $\mathcal{D}=\{(i,j): i\in [n-1],\,j\in [n]\setminus \{i\}\}$ be the arc set of the complete directed graph where only node $n$ has no outgoing arcs.
Let $\mathcal{A}_n \subset \{0,1\}^{(n-1)^2}$ be the set of characteristic vectors of arborescences rooted at $n$.
Here an arborescence, also called a directed spanning tree, is a set of arcs such that every node $i\in[n-1]$ has outdegree one and every vertex can reach node $n$.
The latter condition is equivalent to requiring every nonempty subset $S\subseteq [n-1]$ to have at least one outgoing arc leaving $S$.
The \emph{directed spanning tree polytope} is $P_{\rm DST} = \text{conv}(\mathcal{A}_n)$ with the halfspace description
\[
P_{\rm DST}=\left\{z\in \R_{\geq 0}^\mathcal{D}: \sum_{j\in [n]\setminus\{i\}} z_{ij}=1,\, i\in [n-1],\, \sum_{(i,j)\in A\cap (S\times S)}z_{ij}\leq |S|-1, \emptyset \neq S\subseteq [n-1]\right\}.
\]
The inequalities $\sum_{(i,j)\in \mathcal{D}\cap (S\times S)}z_{ij}\leq |S|-1$ are also called \emph{subtour-elimination} constraints.

For $z \in \mathbb{R}^\mathcal{D}$, we define the node weight of a node $i\in [n-1]$ by $a_i(z) \coloneq \sum_{j\in [n]\setminus \{i\}} z_{ij}$.
The following lemma is analogous to \Cref{lem:matching:fixed_node_weights}.
\begin{lemma}
\label{lem:dst:fixed_node_weights}
Let $(X,Y)$ be a nonempty rectangle of $P_{\rm DST}$.
Then there is a vector $a^*\in [0,1]^{n-1}$ such that, for all $x \in X$, we have $a_i(x) = a^*_i$ for all $i \in [n-1]$.
Consequently, for all $y \in Y$ and $i\in [n-1]$, we have $a_i(y)=1 - a^*_i$.
\end{lemma}

The following lemma shows that the set of arcs that can appear in arborescences contained in a rectangle is restricted.
\begin{lemma}
\label{lem:restricted_dst}
Let $(X,Y)$ be a vertex-realizing rectangle of $P_{\rm DST}$, and let $a^*$ be defined as in \Cref{lem:dst:fixed_node_weights}.
Define
\[
I \coloneq \{i\in [n-1]: a^*_i > 0\},\quad J\coloneq [n-1]\setminus I,
\]
and
\[
E\coloneq \{(i,j): m_{ij}=1 \textrm{ for some } m\in (X+Y) \cap \mathcal{A}_n\}.
\]

Then
\[
|E|\leq (n-1)^2-|I|\cdot (n-1-|I|).
\]
\end{lemma}
\begin{proof}
    We prove by contradiction that for every $i\in I$ and $j\in J$, at most one of the arcs $(i,j)$ and $(j,i)$ belongs to $E$.
    Suppose that there are arborescences $m,m'\in (X+Y) \cap \mathcal{A}_n$ with $m_{ij}=1$ and $m'_{ji}=1$.
    We write $m=x+y$ and $m'=x'+y'$ with $x,x'\in X$ and $y,y'\in Y$.
    Since $i\in I$, we have $\sum_{k\in [n]\setminus \{i\}}x_{ik}=a_i^*>0$.
    Because $0\leq x\leq m$ and $m$ has exactly one outgoing arc from $i$, we must have $x_{ij}>0$.
    Since $j\in J$, we have $\sum_{k\in [n]\setminus \{j\}}x'_{jk}=a_j^*=0$, which implies $x'_{ji}=0$ (since $x'\geq 0)$.
    Hence $y_{ji}'=x_{ji}'+y_{ji}'=m_{ji}'=1$.
    Because $X+Y\subseteq P_{\rm DST}$, we have $x+y'\in P_{\rm DST}$.
    However, it violates the subtour-elimination constraint for $S=\{i,j\}$ 
    \[
    (x+y')_{ij}+(x+y')_{ji}=x_{ij}+y'_{ij}+x_{ji}+y'_{ji}\geq x_{ij}+y_{ji}'=x_{ij}+1>1,
    \]
    which gives a contradiction and proves the claim.
    Thus, at least $|I||J|$ of the $(n-1)^2$ arcs in $A$ are not in $E$, which gives $|E|\leq (n-1)^2-|I|\cdot |J|=(n-1)^2-|I|\cdot (n-1-|I|)$.
\end{proof}

We use the measure
\[
\mu_2:\mathcal{P}^{(n-1)^2}_+\to \R_{\geq 0},
\qquad
\textstyle A\mapsto \max_{a\in V(A)}|\{i\in [n-1]: \sum_{j\in [n]\setminus \{i\}}a_{ij}>0\}|
\]
with $\mu_2(\emptyset)=0$.
Again, the measure $\mu_2$ corresponds to the degree-measure $A\mapsto\max_{v\in V(A)}|\supp(v)|$ after replacing $|\supp(v)|$ with $|\supp(\pi(v))|$ for a linear map $\pi:\R^A\to \R^{n-1}$.

\begin{proposition}\label{prop:measure_dst}
    The function $\mu_2$ is a measure.
    Moreover, for all $A, B\in \mathcal{P}^{(n-1)^2}_+$,
    \[
    \mu_2(\mathrm{conv}(A\cup B))\leq \max(\mu_2(A),\mu_2(B)).
    \]
\end{proposition}
\begin{proof}
    The proof is identical to the proof of \Cref{prop:measure_matching}.
\end{proof}

\begin{lemma}\label{lem:few_dsts_per_rectangle}
    Let $(X,Y)$ be a vertex-realizing rectangle of $P_{\rm DST}$ with $\mu_2(X)\in ((n-1)/3,2(n-1)/3]$.
    Then
    \[
    |(X+Y)\cap \mathcal{A}_n|\leq\left(\frac{7}{9}\right)^{n-1}(n-1)^{n-1}.
    \]
\end{lemma}
\begin{proof}
    By \Cref{lem:dst:fixed_node_weights}, we have $\sum_{j\in [n]\setminus \{i\}}x_{ij}=a_i^*$ for all $x\in X$ and $i\in [n-1]$.
    With $I=\{i\in[n-1]:a^*_i>0\}$, we have $r\coloneq\mu_2(X)=|I|$.
    Since $r\in ((n-1)/3,2(n-1)/3]$, we have $r (n-1-r)\geq \frac{2(n-1)}{3}\cdot \frac{(n-1)}{3}=\frac{2}{9}(n-1)^2$.
    With \Cref{lem:restricted_dst}, it follows that
    \[
    \textstyle
    |E|\leq (n-1)^2-|I|\cdot (n-1-|I|)=(n-1)^2-r\cdot (n-1-r)\leq \frac{7}{9}(n-1)^2.
    \]
    Every arborescence $m\in (X+Y)\cap \mathcal{A}_n$ is obtained by choosing exactly one outgoing arc for each node $i\in [n-1]$, and all chosen arcs must lie in $E$.
    Let $d_i$ be the outdegree of node $i$ in the directed graph induced by $E$.
    Then
    \[
    |(X+Y)\cap \mathcal{A}_n|\leq 
    \prod_{i=1}^{n-1}d_i
    \leq \left(\frac{\sum_{i=1}^{n-1}d_i}{n-1}\right)^{n-1}
    =\left(\frac{|E|}{n-1}\right)^{n-1}
    \leq \left(\frac{7}{9}\right)^{n-1}(n-1)^{n-1}.
    \qedhere
    \]
\end{proof}
\dststclowerbound*
\begin{proof}
    Consider a decomposition 
    \[
    P_{\rm DST} = \mathrm{conv}\left(A_0 \cup \bigcup_{i=1}^k (A_i + B_i)\right)
    \]
    as in \Cref{cor:general_bound} with $r=(n-1)/3$.
    As in the proof of \Cref{thm:matching_lower_bound}, we must have $A_0=\emptyset$.

    There are $n^{n-2}$ arborescences and every arborescence must be contained in at least one rectangle.
    By \Cref{lem:few_dsts_per_rectangle}, one rectangle contains at most $\left(\frac{7}{9}\right)^{n-1}(n-1)^{n-1}$ arborescences.
    Therefore
    \[
    k\geq \frac{n^{n-2}}{\left(7/9\right)^{n-1}(n-1)^{n-1}}
    =\left(\frac{9}{7}\right)^{n-1}\cdot \frac{1}{n}\cdot \left(1+\frac{1}{n-1}\right)^{n-1}\in 2^{\Omega(n)}.
    \]
    Again, the lower bound on $\size_+(\Phi)$ follows from \Cref{cor:general_bound} and \Cref{prop:measure_dst}.
\end{proof}
Dropping polynomial factors, this lower bound is tight up to a multiplicative factor in the exponent.
\begin{proposition}
    There is a tropical circuit of size $\mathcal{O}(n2^n)$ computing $f_{P_{\rm DST}}$.
\end{proposition}
\begin{proof}
    We give an explicit dynamic program.
    For $I\subseteq [n-1]$, let $D(I)$ be the maximum weight of an arborescence on the node set $I\cup \{n\}$ rooted at $n$.
    We set $D(\emptyset)=0$.
    For $i\in [n-1]$ and $J\subseteq [n-1]\setminus \{i\}$, we define 
    \[
    M_i(J)=\max_{k\in J\cup \{n\}}x_{ik}.
    \]
    Then, for $I\neq \emptyset$,
    \[
    D(I)=\max_{i\in I}\left\{D(I\setminus \{i\})+M_i(I\setminus \{i\})\right\}.
    \]
    The output is $D([n-1])$.
    We now prove the correctness of the recurrence.
    For $I=\{i\}$, $D(\{i\})=x_{in}$ and the recurrence is correct.
    Now, let $S$ be an arborescence on $I\cup \{n\}$ (rooted at $n$) with maximum weight $W$ and let $v$ be a leaf of $S$ with $(v,w)\in S$.
    Then, we have $W=(W-x_{vw})+x_{vw}=D(I\setminus \{v\})+M_v(I\setminus \{v\})$ and thus $W\leq D(I)$.
    Every choice $i\in I$ in the recurrence corresponds to an arborescence on $I\cup \{n\}$, since $D(I\setminus\{i\})$ corresponds to an arborescence on $I\setminus \{i\}\cup \{n\}$ and $M_i(I\setminus \{i\})$ corresponds to choosing an arc from $i$ to $I\setminus \{i\}\cup \{n\}$.
    Thus $D(I)\leq W$ and equality follows.
    Hence the recurrence is correct.
    
    The values $M_i(J)$ can be computed with $\max$ gates using $\mathcal{O}(n2^n)$ gates.
    The recurrence $D(I)$ can be implemented with additional $\mathcal{O}(n2^n)$ $+$ and $\max$ gates.
    Thus the total size of the induced tropical circuit is $\mathcal{O}(n2^n)$.
\end{proof}
For the upper bound in \Cref{cor:mnnc_dst_lower_bound}, we use a $(\max, +,-)$-circuit which is the tropicalization of the $(+,\times, /)$-circuit that corresponds to the directed star-mesh transformation \citep{fomin2016subtraction}; see also \citep{hertrich2025relu} for a tropicalization of the undirected version.

\dstseparation*
\begin{proof}
    \textbf{Lower bound.}
    The lower bound follows exactly as in \Cref{cor:mnnc_matching_lower_bound}.
    Let $s$ be the size of a rank-2 maxout ICNN computing $f_{P_{\rm DST}}$.
    Since $f_{P_{\rm DST}}$ is positively homogeneous and monotone, \Cref{lem:bias_to_nonbias_for_pos_hom_functions} and \Cref{lem:ICNN_to_mon_for_mon_functions} give a bias-free monotone rank-2 maxout network of size $s$ computing $f_{P_{\rm DST}}$.
    By \Cref{lem:mon_maxout_to_STC}, this network gives an STC $\Phi$ with $\size(\Phi)\leq 2s^2+2s(n-1)^2$. By \Cref{thm:dst_lower_bound}, $\size(\Phi)\geq 2^{\Omega(n)}$.
    Therefore $2s^2+2s(n-1)^2\geq 2^{\Omega(n)}$, which implies $s\geq 2^{\Omega(n)}$.

    \textbf{Upper bound.}
    There is a $(+,\times, /)$-circuit of size $\mathcal{O}(n^3)$ which computes the basis generating polynomial $\sum_{S\in \mathcal{A}_n}\prod_{(i,j)\in S}x_{ij}$ via the directed star-mesh transformation \citep[Section 7]{fomin2016subtraction}.
    Replacing $+$ gates by $\max$ gates, $\times$ gates by $+$ gates, and $/$ gates by $-$ gates yields a $(\max,+,-)$-circuit of size $\mathcal{O}(n^3)$ that computes $f_{P_{\rm DST}}$ by tropicalization \citep[Proposition 4]{hertrich2026arithmetic}.
    Since a $(\max,+,-)$-circuit of size $s$ directly translates to a maxout network of size $s$, we have $\nnc(P_{\rm DST})\in \mathcal{O}(n^3)$.
    
    For the sake of completeness, we provide a sketch of the corresponding dynamic program here.
    We start with arc weights $x^{(1)}_{ij}=x_{ij}$ for $i\in [n-1]$ and $j\in [n]\setminus \{i\}$.
    For $k=1,\dots, n-1$, assume that the weights $x^{(k)}_{ij}$ are defined on the current vertex set $\{k,k+1,\dots, n\}$.
    Let 
    \[
    y_k=\max_{i\in\{k+1,\dots,n\}}x^{(k)}_{ki}
    \]
    be the maximum weight of any arc that goes out of the node $k$ at step $k$.
    The idea is to delete one node from the current graph and to modify the arc weights of the remaining arcs such that the maximum weight of an arborescence rooted at $n$ of this smaller graph is equal to the weight of the maximum arborescence rooted at $n$ of the original graph.
    The arc weights are modified via the recursion
    \[
    x^{(k+1)}_{ij}=\max\left(x^{(k)}_{ij},x^{(k)}_{ik}+x^{(k)}_{kj}-y_{k}\right)
    \]
    for all distinct $i,j\in \{k+1,\dots, n\}$ with $i\neq n$.
    The output is $y_1+\dots+ y_{n-1}$.
\end{proof}

\subsubsection*{AI Disclosure}
We used large language models (LLMs) solely to assist with writing and presentation of this paper, including suggestions for improving the readability of the paper, language editing, and checking for typographical errors.
All mathematical results in this paper were developed by the authors and not generated by LLMs.
The authors assume responsibility for all content.

\printbibliography
\end{document}